\documentclass[11pt,letterpaper]{article}

\usepackage{fullpage}
\usepackage{amsmath,amssymb,amsthm,dsfont}
\usepackage{tikz}
\usetikzlibrary{shapes.symbols,arrows}

\newtheorem{lem}{Lemma}
\newtheorem{thm}{Theorem}
\newtheorem{cor}{Corollary}
\theoremstyle{definition}

\DeclareMathOperator{\Pa}{{\mathcal{P}}}

\DeclareMathOperator{\lcm}{lcm}
\DeclareMathOperator{\Rem}{Rem}
\DeclareMathOperator{\Simp}{Simp}

\DeclareMathOperator{\Step}{Step}
\DeclareMathOperator{\Red}{Red}

\DeclareMathOperator{\Start}{Start}
\DeclareMathOperator{\End}{End}
\DeclareMathOperator{\Edges}{Edges}

\newcommand{\CDD}{{cd}}
\newcommand{\EP}{{ep}}
\newcommand{\CF}{{cr}}
\newcommand{\nc}{{\mathrm{ nc}}}

\newcommand{\cc}{{\mathrm {c}}}

\newcommand{\IR}{\mathds{R}}
\newcommand{\IN}{\mathds{N}}

\newcommand{\M}{\mathcal{M}}

\newcommand{\Bcnc}{B_{\mathrm {c}}}		
\newcommand{\Bunu}{B_{\mathrm {m/s}}}    	
\newcommand{\Bms}{B_{\mathrm {m/s}}}    	
\newcommand{\Bep}{B_\mathrm{e}}		
\newcommand{\Benp}{B_\mathrm{e}}		
\newcommand{\Bneone}{B_\mathrm{r}}		
\newcommand{\Bnetwo}{B_\mathrm{r}^\mathrm{HA}}	

\newcommand{\SCC}{\mathcal{C}}
\newcommand{\CP}{\mathcal{P}_\circlearrowleft}
\newcommand{\IRmax}{{\IR}_{\max}}
\newcommand{\IRmin}{{\IR}_{\min}}
\newcommand{\wstar}{{w}_{*}}
\newcommand{\ito}{{i\!\to}}

\newcommand{\realrem}[2]{\mathbf{N}_{\geqslant {#1}}^{({#2})}}

\renewcommand{\le}{\leqslant}
\renewcommand{\leq}{\leqslant}
\renewcommand{\ge}{\geqslant}
\renewcommand{\geq}{\geqslant}

\title{On the Transience of Linear Max-Plus Dynamical Systems}
\author{Bernadette Charron-Bost \and Matthias F\"ugger \and Thomas Nowak}
\begin{document}

\date{}

\maketitle
\begin{abstract}
We study the transients of linear max-plus dynamical systems.
For that, we consider for each irreducible max-plus matrix $A$, the weighted graph 
	$G(A)$ such that $A$ is the adjacency matrix of $G(A)$.
Based on a novel graph-theoretic counterpart to the number-theoretic Brauer's theorem, 
	we propose two new methods for the construction of arbitrarily long paths in $G(A)$
	with maximal weight.
That leads to two new upper bounds on the transient of a linear max-plus system
	which both improve on the bounds previously given by Even and Rajsbaum (STOC~1990, Theory of Computing Systems~1997),
	by Bouillard and Gaujal (Research Report~2000), and by Soto y Koelemeijer (PhD Thesis~2003),
	and are, in general, incomparable with Hartmann and Arguelles' bound (Mathematics of Operations Research~1999).
With our approach, we also show how to improve the latter bound  by a factor of two.

A significant benefit of our bounds is that each of them turns out  to be linear 
	in the size of the system in various classes of linear max-plus system whereas the bounds
     previously given are all at least quadratic. 
Our second result concerns the relationship between
     matrix and system transients:  We prove that the transient of an
     $N\times N$ matrix~$A$ is, up to some constant, equal
     to the  transient of an~$A$-linear system with an initial vector
	whose norm is quadratic in~$N$.
Finally, we study the applicability of our results to the well-known Full Reversal algorithm
	whose behavior can be described as a min-plus linear system.
\end{abstract}

\thispagestyle{empty}

\newpage
\tableofcontents
\thispagestyle{empty}
\newpage
\setcounter{page}{1}

\section{Introduction}

The mathematical theory of linear max-plus dynamical systems provides
     tools  to understand the complex behavior of many important
     distributed systems.
Of particular interest are transportation and automated manufacturing
     systems~\cite{GBO98, CDQV85, DS90}, or network
     synchronizers~\cite{even:rajs}.
As shown by Charron-Bost et al.~\cite{fr:sirocco}, another striking example is the {\em
     Full Reversal\/} distributed algorithm which  can be used to solve
     a variety of problems like routing~\cite{GB87} and
     scheduling~\cite{BG89}.
The fundamental theorem in max-plus algebra---an analog of the
     Perron-Frobenius theorem---states that the powers of an
     irreducible max-plus matrix become periodic after a finite index,
     called the {\em transient\/}  of the matrix (see for instance~\cite{workplus}).
As an immediate corollary, any linear max-plus system is periodic
     from some index, called the {\em transient\/} of the (linear)
     system, which clearly depends on initial conditions and  is at
     most equal to the transient of the matrix of the system.

For all the above mentioned applications, the study of the transient
     plays a key role in characterizing their performances: For
     example, in case of Full Reversal routing, the system transient
     corresponds to the time until the routing algorithm terminates in
     a destination oriented routing graph.
Besides that, understanding the matrix and system transient is of
     interest on its own for the theory of max-plus algebra.
While the transients of matrices and linear systems have been shown to
     be computable in polynomial time by Hartmann and
     Arguelles~\cite{hartmann:arguelles}, their algorithms provide no
     analysis of the transient phase, as they both use (binary) search
     at heart, and by that do not hint at the parameters that
     influence matrix and system transients.
Conversely, upper bounds involving these parameters would help to
     predict the duration of the transient phase, and to define
     strategies to reduce the transient as well.
Hence concerning transience bounds, our contribution is {\em both\/} numerical
     and methodological.
	
The problem of bounding the transients has already been studied:
     Bouillard and  Gaujal~\cite{bouillard:gaujal} have given an upper
     bound on the transient of a matrix which is exponential in the
     size of the matrix, and  polynomial bounds have been established
     by Even and Rajsbaum~\cite{even:rajs} for linear systems with
     integer coefficients, by Soto y Koelemeijer~\cite{koelemeijer}
     for both general matrices and linear systems, and
     by Hartmann and Arguelles~\cite{hartmann:arguelles} for general
     matrices and linear systems in max-plus algebra.
In each of these works, the problem of studying the transient is
     reduced to the study of paths in a specific graph: For every
     max-plus matrix~$A$, one considers the weighted directed
     graph~$G$ whose adjacency matrix is~$A$, and the {\em critical
     subgraph\/} of~$G$ which consists of the {\em critical closed
     paths\/} in~$G$, namely those closed paths with maximal average
     weight.
The periodic behavior of the powers of~$A$ is intimately related to
     the structure of the critical subgraph of~$G$: Bounding
     transients amounts to bounding the weights of arbitrary long
     paths in the graph.
The first step in controlling the weights of paths consists in
     reaching the critical  subgraph with sufficiently  long paths.
With respect to this first step, the methods used in the four
     above-mentioned transience  bounds are rather similar.
The approaches mainly differ in the way the critical subgraph is then
     visited.

In this article, we propose two new methods, namely the {\em
     explorative method\/} and the {\em repetitive method\/} for
     visiting the critical subgraph.
The first one consists in exploring the whole strongly connected
     components of the  critical subgraph whereas in the second one,
     the visit of the critical subgraph is confined to repeatedly
     follow only one closed path.
That leads us to two new upper bounds on the transients of linear
     systems which improve on the bounds given by Even and Rajsbaum
     and by Soto y Koelemeijer, and are incomparable with Hartmann and
     Arguelles' bound, for which we show how to improve it by a factor
     of two.

A significant benefit of our bounds lies in the fact  that each
     of them turns out  to be linear in the size of the system 
	 (i.e., the number of nodes) in some
     important graph families (e.g., trees) whereas the bounds
     previously given are all at least quadratic.
This is mainly due to the introduction of new graph parameters that
     enable a fine-grained  analysis of the transient phase.
In particular, we introduce the notion of the {\em exploration
     penalty\/} of a graph~$G$ as the least  integer~$k$ with the
     property that, for every  $n\geqslant k$ divisible by the
     cyclicity of~$G$ and every node~$i$ of~$G$, there is a closed
     path  starting and ending at~$i$ of length~$n$.
One key point is then an at most quadratic upper bound on the
     exploration penalty which we derive from  the number-theoretic
     Brauer's Theorem~\cite{Bra42}.

Another contribution of this paper concerns the relationship between
     matrix and system transients:  We prove that the transient of an
     $N\times N$ matrix~$A$---which is clearly an upper bound on all
     transients of~$A$-linear systems---is, up to some constant, equal
     to the  transient of an~$A$-linear system with an initial vector
     whose norm is quadratic in~$N$.
In addition to shedding new light on transients, this result provides
     a direct method for deriving upper bounds on matrix transients
     from upper bounds on system transients.

The paper is organized as follows.
Section~\ref{sec:prelim} introduces basic notions of graph theory and max-plus algebra.
We show an upper bound on lengths of maximum weight paths that do not visit the critical subgraph in Section~\ref{sec:visiting}.
In Section~\ref{sec:explorationpenality}, we introduce the notion of {\em exploration penalty\/} and improve a theorem by Denardo~\cite{denardo} on the existence of arbitrarily long paths in strongly connected graphs.
Sections~\ref{sec:explorative} and~\ref{sec:nonexplorative} introduce our explorative and repitive bounds, respectively.
We show how to convert upper bounds on the transients of max-plus systems to upper bounds on the transients of max-plus matrices in Section~\ref{sec:matrix}.
We discuss our results, by comparing them to previous work and by applying them to the analysis of the Full Reversal algorithm, in Section~\ref{sec:discussion}.

\section{Preliminaries}\label{sec:prelim}

\subsection{Basic definitions}

Denote by $\IN$ the set of nonnegative integers and let $\IN^* = \IN
     \setminus \{0\}$.
Let $\IRmax = \IR \cup \{-\infty\}$.
In this paper, we follow the convention $\max\emptyset=-\infty$.

A {\em (directed) graph\/}~$G$ is a pair~$(V,E)$, where~$V$ is
     a nonempty finite set and $E\subseteq V\times V$.
We call the elements of~$V$ the {\em nodes\/} of~$G$ and the elements
     of~$E$ the {\em edges\/} of~$G$.
An edge $e=(i,j)$ is called {\em incident to\/}
     node $k$ if $k = i$ or $k = j$.
We say that~$G$ is {\em nontrivial\/} if~$E$ is nonempty.

A graph $G'=(V',E')$ is a {\em subgraph\/} of~$G$ if $V'\subseteq V$
     and $E' \subseteq E$.
For a nonempty subset~$E'$ of~$E$, let the {\em subgraph  of $G$ induced by edge set $E'$\/} be the
     graph~$(V',E')$ where $V' = \{i
     \in V \mid \exists j\in V: (i,j)\in E' \vee (j,i)\in E'\}$.

A {\em path\/}~$\pi$ in $G$ is a triple $\pi=(\Start,\Edges,\End)$
     where $\Start$ and $\End$ are nodes in $G$, $\Edges$ is a
     sequence $(e_1,e_2,\dots,e_n)$ of edges  $e_l=(i_l,j_l)$ such
     that $j_l=i_{l+1}$ if $1\leqslant l\leqslant n-1$ and, if
     $\Edges$ is nonempty, $i_1=\Start$ and $j_n=\End$, and if
     $\Edges$ is empty, $\Start=\End$.
We say that~$\pi$ is {\em empty\/} if $\Edges$ is empty.
We define the operators $\Start$, $\Edges$, and $\End$ on the set of
     paths by setting $\Start(\pi)=\Start$, $\Edges(\pi)=\Edges$, and
     $\End(\pi)=\End$.
Define the {\em length\/}~$\ell(\pi)$ of~$\pi$ as the length of $\Edges(\pi)$.

Let $\Pa(i,j,G)$ denote the set of paths $\pi$ in $G$ with
     $\Start(\pi)=i$ and $\End(\pi)=j$, and by $\Pa(\ito,G)$ the set of
     paths~$\pi$ in~$G$ with $\Start(\pi)=i$.
If $\pi\in\Pa(i,j,G)$, we say that $i$ is the {\em start node\/} of
     $\pi$ and~$j$ is the {\em end node\/} of $\pi$.
We write $\Pa^n(i,j,G)$ (respectively $\Pa^n(\ito,G)$), where $n\ge
     0$, for the set of paths in $\Pa(i,j,G)$ (respectively $\Pa(\ito
     ,G)$) of length $n$.
Path~$\pi$ is {\em closed\/} if $\Start(\pi)=\End(\pi)$.
Let $\CP(G)$ denote the set of {\em nonempty\/} closed paths in graph~$G$.

Path~$\pi$ is {\em elementary\/} if the nodes~$i_l$ are pairwise distinct.
Intuitively, an elementary path visits at most one node twice: its start node.
All elementary paths~$\pi$ satisfy $\ell(\pi)\leqslant \lvert V\rvert$.
A path is {\em simple\/} if it is elementary and non-closed, or empty.
Intuitively, a path is simple if it does not visit the same node twice.
All simple paths~$\pi$ satisfy $\ell(\pi) \leqslant \lvert V\rvert -1$.

Call $G$ {\em strongly connected\/} if $\Pa(i,j,G)$ is nonempty for
     all nodes $i$ and $j$ of $G$.
A subgraph $H$ of $G$ is called a {\em strongly connected
     component\/} of $G$ if $H$ is maximal with respect to the
     subgraph relation such that $H$ is strongly connected.
Denote by $\SCC(G)$ the set of strongly connected components of~$G$.
For every node $i$ of $G$ there exists exactly one $H$ in $\SCC(G)$
     such that $i$ is a node of $H$.

For two paths $\pi$ and $\pi'$ in $G$, we say that $\pi'$ is a {\em
     prefix of $\pi$\/} if $\Start(\pi)=\Start(\pi')$ and
     $\Edges(\pi')$ is a prefix of $\Edges(\pi)$.
We say that $\pi'$ is a {\em postfix of $\pi$\/} if
     $\End(\pi)=\End(\pi')$ and $\Edges(\pi')$ is a postfix of
     $\Edges(\pi)$.
We call $\pi'$ a {\em subpath of $\pi$\/} if it is the postfix of some
     prefix of $\pi$.
We say a node~$i$ is a {\em node of path $\pi$\/} if there exists a
     prefix $\pi'$ of $\pi$ with $\End(\pi')=i$, and an edge $e$ is an
     {\em edge of path $\pi$\/} if $e$ occurs within $\Edges(\pi)$.
For two paths $\pi_1$ and $\pi_2$ with $\End(\pi_1) = \Start(\pi_2)$,
     define the concatenation $\pi = \pi_1 \cdot \pi_2$ by setting
     $\Start(\pi) = \Start(\pi_1)$, $\Edges(\pi) = \Edges(\pi_1) \cdot
     \Edges(\pi_2)$ and $\End(\pi) = \End(\pi_2)$.

We define the {\em girth\/}~$g(G)$ and {\em circumference\/}~$\CF(G)$ of~$G$ to be the minimum resp.\ maximum length of
     nonempty elementary closed paths in $G$.
Define the {\em cab driver's diameter\/}~$\CDD(G)$ of~$G$ to be the
     maximum length of simple paths in $G$.
The {\em cyclicity\/} of~$G$ is defined as 
\begin{equation*}
  c(G)  = \lcm\big\{ \gcd\{ \ell(\gamma)  \mid \gamma \in \CP(H) \}  \mid H \in \SCC(G) \big\}\enspace.
\end{equation*}
For a strongly connected graph~$G$, the cyclicity~$c(G)$ is the greatest common divisor of closed path lengths in~$G$; in particular~$c(G)\leqslant g(G)\leqslant \lvert V\rvert$.
We say that~$G$ is {\em primitive\/} if $c(G)=1$.
A graph is primitive if and only if all its strongly connected components are primitive.
We also define the two graph parameters
\[ d(G) = \max \{ c(H) \mid H\in \SCC(G) \} \]
and
\[ p(G) = \lcm \{ \ell(\gamma) \mid \gamma \in \CP(H) \wedge \gamma \text{ is elementary}\}\enspace. \]
It is $d(G)\leqslant c(G)\leqslant p(G)$.

\begin{lem}\label{lem:path:split}
Let $\pi$ be a path in $G$.
Then there exist paths $\pi_1$, $\pi_2$ and an elementary closed
     path $\gamma$ such that $\pi=\pi_1\cdot\gamma\cdot\pi_2$.
Moreover, if $\pi$ is non-simple, then $\gamma$ can be chosen to be nonempty.
\end{lem}
\begin{proof}
The first claim is trivial, for we can choose~$\pi_1=\pi$, and~$\gamma$ and~$\pi_2$ to be empty.

So let~$\pi$ be non-simple.
Then~$\pi$ is necessarily nonempty.
If $\pi$ is elementary closed, set $\gamma=\pi$ and choose $\pi_1$ and $\pi_2$ to
     be empty.
Otherwise, let $\Edges(\pi)=(e_1,\dots,e_n)$ and $e_l=(i_l,j_l)$.
By assumption, there exist $k<l$ such that $i_k=i_l$ and $l-k$ is
     minimal.
Now set $\Edges(\pi_1)=(e_1,\dots,e_{k-1})$,
     $\Edges(\gamma)=(e_k,\dots,e_{l-1})$, $\Edges(e_l,\dots,e_n)$,
     and choose the start and end nodes accordingly.
Closed path $\gamma$ is elementary, because $l-k$ was chosen to be
     minimal.
\end{proof}

An {\em edge-weighted\/} (or {\em e-weighted\/}) graph~$G$ is a
triple~$(V,E,w_E)$ such that~$(V,E)$ is a graph and $w_E:E\to\IR$.
An {\em edge-node-weighted\/} (or {\em en-weighted\/}) graph~$G$ is a quadruple~$(V,E,w_E,w_V)$
     such that~$(V,E,w_E)$ is an e-weighted graph and $w_V:V\to\IRmax$.

If $G=(V,E,w_E,w_V)$ is an en-weighted graph and~$\pi$ is a path
in~$G$ with $\Edges(\pi) = (e_1,\dots,e_n)$,
we define the {\em en-weight\/} of~$\pi$ as
\[ w(\pi) = \sum_{l=1}^n w_E(e_l) + w_V\big( \!\End(\pi) \big) \enspace, \]
and, if $G$ is just an e-weighted graph, the {\em e-weight\/} of~$\pi$ as
\[ \wstar(\pi) = \sum_{l=1}^n w_E(e_l)\enspace. \]
Let $w^n(\ito,G) = \max\{ w(\pi) \mid \pi\in \Pa^n(\ito,G) \}$.

In the rest of the paper, every notion introduced for graphs (resp.\
e-weighted graphs) is trivially extended to e-weighted graphs (resp.\
en-weighted graphs) using the same terminology and notation.

\subsection{Realizers}

Let~${\mathbf N}$ be any nonempty set of nonnegative integers,
	and let~$i$ be any node of a strongly connected e-weighted 
	graph $G$.
A path $\tilde{\pi}$ is said to be an  ${\mathbf N}$-{\em realizer for node} $i$
	if $\tilde{\pi} \in  \Pa(\ito,G)$,  $\ell(\tilde{\pi}) \in {\mathbf N}$,
	and  $w(\tilde{\pi}) =  \sup_{n\in {\mathbf N}} w^n(\ito)$.
Obviously, an ${\mathbf N}$-realizer exists for any node if ${\mathbf N}$ is
	finite.
Of particular interest is the case of sets ${\mathbf N}$ of the form
	$$\realrem{\hat{n}}{r,p} = 
	\{ n\in \IN \mid n \geqslant \hat{n}\  \wedge \ n\equiv r\pmod p \}\enspace,$$
	where $r\in \IN$, and $p, \hat{n}\in \IN^*$.
The following lemma gives a useful sufficient condition in terms of
	$\realrem{\hat{n}}{r,p}$-realizer to guarantee the eventual periodicity of 
	the sequence $\big(w^n(\ito )\big)_{n\in \IN}$.
	
\begin{lem}\label{lem:final:step}
Let~$i$ be a node and let~$p$ and~$\hat{n}$ be positive integers.
Suppose that for all~$n\geqslant \hat{n}$, there exists a $\realrem{\hat{n}}{n,p}$-realizer 
	for~$i$ of length $n$.
Then $w^{n+p}(\ito ) = w^n(\ito )$ for all~$n\geqslant \hat{n}$.
\end{lem}
\begin{proof}
For each integer $n\geqslant \hat{n}$, let $\pi_n$ be one $\realrem{\hat{n}}{n,p}$-realizer 
		for~$i$ of length $n$.
Denote by~$X(n)$ the set of paths~$\pi$ in~$\Pa(\ito)$ that satisfy
	$\ell(\pi)\equiv n\pmod p$ and $\ell(\pi)\geqslant \hat{n}$, and
	    let~$x(n)$ be the supremum of values~$w(\pi)$ where~$\pi\in
	    X(n)$.
Since $\realrem{\hat{n}}{n,p}$ has a realizer, each $x(n)$ is  finite, and $ x(n) = w(\pi_n)$.

From $n+p \equiv n\pmod p$, it follows $$X(n+p) = X(n)$$ and so $x(n+p)=x(n)$.
For all~$n\geqslant \hat{n}$, we have~$\Pa^n(\ito) \subseteq X(n)$,
	    i.e., $w^n(\ito ) \leqslant x(n)$.
As $n +p > n$, we also conclude that  $w^{n+p}(\ito) \leqslant x(n+p)$. 
Because~$\pi_n\in \Pa^n(\to j)$, we have~$w^n(\to j) \geqslant w(\pi_n)$.
Similarly,~$w^{n+p}(\to j) \geqslant w(\pi_{n+p})$. 
This concludes the proof.
	\end{proof}

\subsection{The critical subgraph}

Let~$G = (V,E,w_E)$ be a nontrivial strongly connected e-weighted graph.
Define the {\em rate\/}~$\varrho(G)$ of~$G$ by
\[ \varrho(G) = \sup \left\{ \frac{\wstar(\gamma)}{\ell(\gamma)} \mid \gamma \in \CP(G) \right\} \enspace, \]
which is easily seen to be finite.
A (nonempty) closed path $\gamma \in \CP(G)$ is {\em critical\/} if
     $\wstar(\gamma)/\ell(\gamma) = \varrho(G)$.
A node of $G$ is {\em critical\/} if it is node of a critical path
     in~$G$, and an edge of $G$ is {\em critical\/} if it is an edge
     of a critical path in~$G$.
The {\em critical subgraph\/} of~$G$, denoted by~$G_\cc$, is
     the subgraph of $G$ induced by the set of critical edges
     of~$G$.
A {\em critical component\/} of~$G$ is a strongly connected component of the critical subgraph of~$G$.

We denote by $\Delta(G)$ (respectively $\delta(G)$) the maximum (respectively minimum) edge weights in~$G$.  
Let~$\Delta_\nc(G)$ denote the maximum weight of edges between two non-critical nodes.
If no such edge exists, set $\Delta_\nc(G) = \varrho(G)$.
Note that $\delta(G)\leqslant \varrho(G)\leqslant \Delta(G)$ always holds.
Denote the {\em non-critical rate of $G$} by
	\[ \varrho_{nc} (G) = 
	\sup \left\{ \frac{\wstar(\gamma)}{\ell(\gamma)} \mid \gamma \in \CP(G) 
	\ \wedge \ \gamma \mbox{ has no critical node } \right\} \enspace, \]
	with the classical convention that $\varrho_{nc}(G)= -\infty$.
	if no such path exists.

We now study how the critical graph and the various parameters introduced above are modified
	by homotheties.
Let $\lambda$ be any element in $\IR$, and let $\lambda \! \otimes \! G$ denote the e-weighted graph 
	$ \lambda \! \otimes \! G= (V,E,\lambda \! \otimes \! w_E)$ where for any edge $e \in E$
	\[ \lambda \! \otimes \! w_E(e) = w_E(e) + \lambda \enspace. \]

\begin{lem}\label{lem:hom}
The  e-weighted graph~$\lambda \! \otimes \! G$ has the same critical
	subgraph as $G$, and its rate is equal to 
	$ \varrho (\lambda \! \otimes \! G) = \varrho (G) + \lambda $.
\end{lem}
\begin{proof}
If~$\wstar(\pi)$ and $\lambda \! \otimes \! w_*(\pi)$ denote the respective e-weights of path~$\pi$ in~$G$ 
	and~$\lambda \! \otimes \! G$, then we have the equality
	$\lambda \! \otimes \! \wstar(\pi) / \ell(\pi) = \wstar(\pi) / \ell(\pi) + \lambda$, which 
	implies~$\varrho(\lambda \! \otimes \! G)= \varrho (G) + \lambda$.
The equality $ \big(\lambda \! \otimes \! G\big)_c = G_\cc$ now easily follows.
\end{proof}	

From Lemma~\ref{lem:hom}, we easily check that  each of the parameters 
	$\Delta(G) - \varrho(G)$, $\Delta_\nc(G) - \varrho(G)$, 
	and $\delta(G) - \varrho(G)$ is invariant under homothety:

\begin{lem}\label{lem:invariance}
Let $G$ be an e-weighted graph, and $\lambda$ an element in $\IR$.
If $\lambda \! \otimes \! G$ denotes the e-weighted graph obtained by adding $\lambda$ to each 
	edge weight of $G$, then 
	$\Delta(\lambda \! \otimes \! G)-\varrho(\lambda \! \otimes \! G) = \Delta(G) - \varrho(G)$,
	$\Delta_\nc(\lambda \! \otimes \! G)-\varrho(\lambda \! \otimes \! G) = \Delta_\nc(G) - \varrho(G)$, and
	$\delta(\lambda \! \otimes \! G)-\varrho(\lambda \! \otimes \! G) = \delta(G) - \varrho(G)$.
\end{lem}

We next show that critical closed paths exist by showing that $\varrho(G)$ is equal to the supremum of the finite set
     of $\wstar(\gamma)/\ell(\gamma)$ where $\gamma$ is a nonempty {\em elementary\/} closed path.
Define 
\[ \varrho_e(G) = \sup \left\{ \frac{\wstar(\gamma)}{\ell(\gamma)} \mid \gamma \in \CP(G) \ \wedge \ \gamma\text{ is elementary} \right\} \enspace. \]

\begin{lem}\label{lem:mean}
Let $a,b,c,d$ be real numbers, $b$ and $d$ positive, such that
     $a/b\leqslant c/d$.
Then:   
\begin{equation*}
  \frac{a}{b}\leqslant \frac{a+c}{b+d}\leqslant \frac{c}{d}
\end{equation*}
\end{lem}
\begin{proof}
We have $ad\leqslant bc$ and thus
$$\frac{a+c}{b+d} = \frac{1}{b}\cdot\frac{ab+bc}{b+d}\geqslant \frac{1}{b} \cdot \frac{ab+ad}{b+d} = \frac{a}{b}\enspace.$$
Analogously,
$$\frac{a+c}{b+d} = \frac{1}{d}\cdot\frac{ad+cd}{b+d}\leqslant \frac{1}{d} \cdot \frac{bc+cd}{b+d} = \frac{c}{d}\enspace,$$
which concludes the proof.
\end{proof}

\begin{lem}\label{lem:xi:is:minimum}
For every nontrivial strongly connected e-weighted graph $G$, 
	$$\varrho(G)=\varrho_e(G) \enspace.$$
In particular, there exists an elementary critical closed path in~$G$.
\end{lem}
\begin{proof}
Obviously, $\varrho_e(G)\leqslant \varrho(G)$.

Conversely, we show by induction on $\ell(\gamma)$ that
     $\wstar(\gamma)/\ell(\gamma)\leqslant \varrho_e(G)$ for all nonempty
     closed paths~$\gamma$ in~$G$.
The case $\ell(\gamma)=1$ is trivial, because every closed path of
     length~$1$ is elementary.
Now let $\ell(\gamma)>1$.
If $\gamma$ is elementary, we are done by the definition of
     $\varrho_e(G)$.
If $\gamma$ is non-elementary, it is non-simple.
Thus, by Lemma~\ref{lem:path:split}, there exist
     $\pi_1$, $\pi_2$, and an elementary nonempty closed path~$\gamma'$ such
     that $\gamma=\pi_1\cdot \gamma' \cdot \pi_2$.
It is $\End(\pi_1)=\Start(\pi_2)$, hence $\gamma''=\pi_1\cdot\pi_2$ is
     a closed path.
Furthermore, $\ell(\gamma')<\ell(\gamma)$ because $\gamma'\neq\gamma$,
     and $\ell(\gamma'')<\ell(\gamma)$ because $\gamma'$ is nonempty.

We obtain
$$\frac{\wstar(\gamma)}{\ell(\gamma)} = \frac{\wstar(\gamma')+\wstar(\gamma'')}{\ell(\gamma')+\ell(\gamma'')} \leqslant %
\max\!\left\{\frac{\wstar(\gamma')}{\ell(\gamma')}\ ,\ \frac{\wstar(\gamma'')}{\ell(\gamma'')} \right\} \leqslant \varrho_e(G)$$
by Lemma~\ref{lem:mean} and the induction hypothesis.
Thus we have shown $\varrho(G)\leqslant \varrho_e(G)$.

The last statement follows because the set of nonempty elementary closed
     paths is finite.
\end{proof}

The following is a well-known fact in max-plus algebra:

\begin{lem}\label{lem:paths:in:crit:comps:are:critical}
Let~$G$ be a nontrivial strongly connected e-weighted graph.
Then every nonempty closed path in $G_c$ is critical in~$G$.
\end{lem}
\begin{proof}
Set $\varrho = \varrho(G)$.
Let $\pi$ be a nonempty closed path in $G_c$ with $\Edges(\pi) = (e_1,
     \dots e_n)$.
By definition of $G_c$, all edges $e_l = (i_l,j_l)$, $1\le l\le n$,
     are critical, i.e., there exists a path $\gamma_l \in \CP(G)$
     with
\begin{equation}\label{eq:gamma_k}
  \wstar(\gamma_l)/\ell(\gamma_l) = \varrho\enspace,
\end{equation}
and whose last edge is~$e_l$.
For each $l$, $1\le l\le n$, construct $\gamma'_l$ from $\gamma_l$ by
     removing its last edge.
Thus $\gamma'_l$ is a path with $\Start(\gamma'_l)=j_l$ and
     $\End(\gamma'_l)=i_l$.
Concatenation of the paths $\gamma'_l$ yields a closed path $\gamma' =
     \gamma'_n \cdot \gamma'_{n-1} \cdot \dots \cdot \gamma'_1$.
From \eqref{eq:gamma_k} follows
\begin{align}
  \wstar(\gamma'_l)+w_E(e_l) &= \varrho\,\ell(\gamma'_l)+\varrho\enspace.\notag\\
\intertext{Summing over all $l$ yields}
  \wstar(\gamma')+\wstar(\pi) &= \varrho\,\ell(\gamma')+\varrho\, \ell(\pi)\enspace.\notag\\
\intertext{Combination with $\varrho\,\ell(\gamma')\geq
\wstar(\gamma')$ gives}
  \wstar(\pi) &\geq \varrho\,\ell(\pi)\enspace,\notag
\end{align}
hence~$\pi$ is critical.
\end{proof}

\subsection{Linear max-plus systems}

A matrix with entries in $\IRmax$ is called a {\em max-plus\/} matrix.
We denote by $\M_{M,N}(X)$ the set of $M\times N$ matrices with entries in~$X$.
If $A\in\M_{M,N}(\IRmax)$ and $B\in\M_{N,Q}(\IRmax)$, define $A\otimes
     B \in \M_{M,Q}(\IRmax)$ by setting \[(A\otimes
     B)_{i,j} = \max\{ A_{i,k} + B_{k,j} \mid 1\leqslant k\leqslant N
     \}\enspace.\]
The identity element in the monoid $\big(\M_{N,N}(\IRmax),\otimes\big)$, denoted 
	$[ 0 ]_N$, is the matrix whose
     diagonal entries are equal to~$0$ and all other entries are equal
     to~$-\infty$.
More generally, for any $\lambda\in \IRmax$, the matrix	 whose
	     diagonal entries are equal to~$\lambda$ and all other entries are equal
	     to~$-\infty$ is  denoted by $[ \lambda ]_N$.
If $A\in\M_{N,N}(\IRmax)$, define 
	$A^{\otimes 0}= [ 0 ]_N$, and $A^{\otimes n} = A\otimes A^{\otimes n-1}$ for
     $n\geqslant1$.
For convenience, the matrix $ [ \lambda ]_N \otimes A$
	is simply written  $\lambda \otimes A$; more generally, for any positive integer $n$, 
	the matrix 
	$\big( [ \lambda ]_N \big)^{\otimes  n }\otimes A$, whose $(i,j)$-entry is
	$A_{i,j} + \lambda n$, is  simply written 
	$ \lambda^{\otimes  n }\otimes A$.

For a matrix $A \in \M_{N,N}(\IRmax)$ and a vector $v \in \IRmax^N$, we define the {\em
     linear max-plus system\/}~$x_{A,v}$ by setting     
\begin{align}
x_{A,v}(n) = \begin{cases}
               v & \text{for } n = 0\\
               A \otimes x_{A,v}(n-1) & \text{for } n \geqslant1\enspace.
       \end{cases}\label{eq:x}
\end{align}
Clearly $x_{A,v}(n) = A^{\otimes n}\otimes v$.

Denote by $G(A)$ the
     e-weighted graph with nodes $\{1,2,\dots,N\}$ containing an edge
     $(i,j)$ if and only if $A_{i,j}$ is finite, and set its edge weight
     $w_E(i,j)=A_{i,j}$.
Observe that with the notation in the previous Section, 
	$G(\lambda \otimes A) = \lambda \! \otimes \! G(A)$.

Denote by $G(A,v)$ the en-weighted graph defined in the same way as~$G(A)$ with node weights set to~$w_V(i)=v_i$.
Call $A$ {\em irreducible\/} if $G(A)$ is strongly connected and nontrivial.

The following correspondence between the sequence $(A^{\otimes
     n}(n))_{n\ge 0}$ (respectively the sequence $x_{A,v}$)
     and weights of paths in $G(A)$ (respectively $G(A,v)$) holds:  
\begin{lem}\label{lem:path:formula}
For all matrices $A \in \M_{N,N}(\IRmax)$ and all $v\in\IRmax^N$,
\begin{align}
  (A^{\otimes n})_{i,j} &= \max\left\{ \wstar(\pi) \mid \pi\in\Pa^n\big(i,j,G(A)\big) \right\}\text{ and}\notag\\
  (A^{\otimes n}\otimes v)_i &= \max\left\{ w(\pi) \mid \pi\in\Pa^n\big(\ito,G(A,v)\big) \right\}
        = w^n(\ito,G(A,v))\enspace.\notag
\end{align}
\end{lem}
\noindent We simply denote $(A^{\otimes n})_{i,j}$ by $A^{\otimes n}_{i,j}$ in the sequel.

We define the {\em cyclicity\/} of~$A$, denoted by~$c(A)$, as the cyclicity of the critical 
	subgraph of~$G(A)$,~i.e, \[c(A) = c\big(G_c(A)\big)\enspace.\]
The  main result on  max-plus matrices is an analog of the  Perron-Frobenius theorem in linear algebra.

\begin{thm}[{\cite[Theorems~2.9 and~3.9]{workplus}}]\label{thm:perron}
An irreducible matrix $A\in\M_{N,N}(\IRmax)$ has exactly one eigenvalue $\varrho(A)$
			equal to the rate of the e-weighted graph $G(A)$, i.e., 
			$$ \varrho(A) = \varrho\big(G(A)\big) \enspace.$$
		Moreover,  there exists an integer~$\hat{n}$ such that for every $n\geqslant \hat{n}$:
		\begin{equation*}\label{eq:a:hoch}
		A^{\otimes n+c(A)} = \big( \varrho(A)\big)^{\otimes  c(A) } \otimes A^{\otimes n} \enspace.
		\end{equation*}
\end{thm}
Using the classical operators ``$+$'' and ``$\cdot$'', the above equality  is equivalent to:
	$$\forall i,j\in \{1,\cdots, N\}:\quad
		A^{\otimes n+c(A)} _{i,j} = A^{\otimes n} _{i,j} + c(A)\varrho(A) \enspace. $$

\subsection{Eventually periodic sequences}\label{subsec:period}

Let $X$ be any nonempty set.
A sequence~$f:\IN\to\IRmax^X$ is {\em
     eventually periodic\/} if there exist $p\in \mathds{N}^*$, $w_p
     \in \IRmax$, and $n_p \in \mathds{N}$ such
     that         
$$\forall n \geq n_p:\quad f(n + p) = f(n) + w_p\enspace,$$%
	where $w_p$ stands for the constant function that maps any element in $X$
	to $w_p$.
Such an integer $p$ is called an {\em eventual period\/} (or for short
     a {\em period\/}) of~$f$.
Theorem~\ref{thm:perron} shows that~$c(A)$ is a period of both
     sequences $\big(A^{\otimes n}\big)_{n\geqslant0}$ and~$\big(x_{A,v}(n)\big)_{n\geqslant0}$.

We denote by ${\mathbf P}_f$ the set of periods of $f$.
Clearly ${\mathbf P}_f$ is a nonempty subset of $\mathds{N}$ closed under
     addition.
Let  $p_0 = \min {\mathbf P}_f$ be the minimal period of $f$; hence
     $p_0\mathds{N}^* \subseteq {\mathbf P}_f$.
As $p_0\in {\mathbf P}_f$, there exist $w_0 \in \IRmax$ and $n_0 \in \mathds{N}$ such that     
$$\forall n \geq n_0:\quad f(n + p) = f(n) + w_0\enspace.$$
Let $p$ be any period of $f$, and $p= ap_0 +b$ the Euclidean division of $p$ by $p_0$.
For any integer $n\geq \max \{n_p ,n_0 -b\}$, 
	$$f(n+p)=f(n) + w_p = f(n + b) +aw_0\enspace.$$
It follows that either $b=0$ or $b$ is a period of $f$.
Since $b\leq p_0-1$ and $p_0$ is the smallest period of~$f$, we have $b=0$,
	i.e., $p_0$ divides $p$.
Then we derive that ${\mathbf P}_f\subseteq p_0 \IN^*$.

For any period $p= q p_0$ of $f$, let $n_q$ be the smallest positive integer such that 
	$$\forall n \geq n_q:\quad  f(n+q p_0)=f(n) + q w_0 \enspace.$$
The integer $n_q$ is called the {\em transient of $f$ for period $p$}.

\begin{lem}
For any positive integer $q$, $n_q = n_1$.
\end{lem} 
\begin{proof}
Since for any $n \geq n_1$,
	$$f(n + q p_0) = f(n) + q w_0,$$
	we have $n_q \leq n_1$.
	
We now prove that  $n_q = n_1$ by induction on $q \in \mathds{N}^*$.
\begin{enumerate}
	\item The base case $q=1$ is trivial.
	\item Assume that $n_q = n_1$.
	For any integer $n \geq n_{q+1}$, $$f(n + (q + 1 ) p_0) = f(n) + (q + 1 ) w_0.$$
	Moreover, if $n + p_0 \geq n_q$ then $$ f(n + (q + 1 ) p_0) = f(n + p_0) + p w_0.$$
	It follows that for any integer $n \geq \max \{ n_q - p_0, n_{q+1} \}$, 
	$$f(n + p_0) = f(n) + w_0.$$
	Hence $n_1 \leq \max \{ n_q - p_0, n_{q+1} \}$, and by inductive assumption 
	$n_1 \leq \max \{ n_1 - p_0, n_{q+1} \}$.
	Then we derive $n_1 \leq n_{q+1}$, and so $n_1 = n_{q+1}$ as required.\qedhere
\end{enumerate}	
\end{proof}


It follows that the transient for period~$p$ of a periodic
     sequence is independent of~$p$.
We hence simply call it the {\em transient of $f$}.

Theorem~\ref{thm:perron} states that~$c(A)=c\big(G_c(A)\big)$ is a period of the sequence~$A^{\otimes n}$.
Because~$p(A)=p\big(G(A)\big)$ is a multiple of~$c(A)$, also~$p(A)$ is a period.
We find it more convenient in our proofs to consider~$p(A)$ as the period instead of the period~$c(A)$ suggested by Theorem~\ref{thm:perron}.

Let $A\in\M_{N,N}(\IRmax)$ be an irreducible matrix and let $v\in\IRmax^N$.
We call the transient of the sequence $\big(A^{\otimes n}\big)_{n\geqslant0}$
     the {\em transient of matrix~$A$}, denoted by~$n_A$, and
     the transient of the sequence~$\big(x_{A,v}(n)\big)_{n\geqslant0}$ the {\em transient of system~$x_{A,v}$},
     denoted by $n_{A,v}$.
Obviously, $n_A$ is an upper bound on the transient of the $x_{A,v}$'s,~i.e.,
	$$\sup\big\{ n_{A,v} | v \in \IRmax^N \big\} \leqslant n_A \enspace.$$
Conversely, the equalities 
	\[ A_{i,j}^{\otimes n} = \big( A^{\otimes n} \otimes e^{j} \big)_i \]
	where $e^{j}_i = 0$ if~$i=j$ and~$e^{j}_i=-\infty$ otherwise, show that
	$$\max \big\{ n_{A,e^{j}} | j \in \{1,\cdots,N\} \big\} \geqslant n_A \enspace.$$
Hence, 
	\begin{equation}\label{eq:nA}
		\sup\big\{ n_{A,v} | v \in \IRmax^N \big\} = n_A \enspace.
		\end{equation}

The transients are invariant under homotheties:
 
\begin{lem}\label{lem:transient:invariance}
For all irreducible matrices~$A\in\M_{N,N}(\IRmax)$, all vectors~$v\in\IRmax^N$,
	and all $\lambda \in \IR$,
	we have the equalities of transients $ n_{\lambda \otimes A} = n_A $ 
	and $n_{\lambda \otimes A,v} = n_{A,v}$.
\end{lem}

\begin{proof}
	The lemma follows immediately from the equalities $$\big(\lambda \otimes A \big)^{\otimes n} = 
		\lambda^{\otimes n} \otimes A^{\otimes n}$$ and 
		$$\big(\lambda \otimes A \big)^{\otimes n} \otimes v = 
		\lambda^{\otimes n} \otimes \big(A^{\otimes n} \otimes v\big) \enspace.$$
\end{proof}

\subsection{Reduction to the case of a zero rate}\label{subsec:prelim:zero:rate}

Let~$G=(V,E,w_E)$ be a nontrivial strongly connected e-weighted graph.
Since $\varrho(G) \in \IR$, we may define the e-weighted 
	graph~$\overline{G} = (-\varrho(G))\!\otimes\! G$. 
By Lemma~\ref{lem:hom}, $G$ and $\overline{G}$ have the same critical subgraph, and
	$$\varrho(\overline{G}) = 0 \enspace.$$
Similarly, for any irreducible max-plus matrix $A$, 
	we denote $\overline{A} = (-\varrho(A))\!\otimes\! A$, and we have
	$$\varrho(\overline{A}) = 0 \enspace.$$
Moreover, Lemma~\ref{lem:invariance} gives:
	$$\Delta \big(\overline{G}\big) = \Delta(G) - \varrho(G),\  
	\Delta_\nc \big(\overline{G}\big) =  \Delta_\nc(G) - \varrho(G),\  
	\delta \big(\overline{G}\big) =  \delta(G) - \varrho(G) \enspace,$$
	which are respectively denoted $\overline{\Delta}(G)$, $\overline{\Delta}_\nc(G)$,
	and $\overline{\delta}(G)$.
From $\varrho(\overline{G}) = 0$, we easily deduce that
	$$ \overline{\delta}(G) \leqslant 0 \leqslant \overline{\Delta}(G) \enspace .$$
The point of the reduction to a zero rate is evidenced by the following lemma:

\begin{lem}\label{lem:realizers:exist}
Let~${\mathbf N}$ be any nonempty set of nonnegative integers,
	and let~$i$ be any node of a strongly connected e-weighted 
	graph $G$ such that $\varrho(G)=0$.
Then there exists an ${\mathbf N}$-realizer for node $i$.
\end{lem}
\begin{proof}
By Theorem~\ref{thm:perron} and the fact that~$\varrho(G)=0$, the
	supremum is taken over a finite set.
Hence it is a maximum.
\end{proof}

\section{Visiting the Critical Subgraph along Optimal Paths}\label{sec:visiting}

In this section, we give an upper bound~$\Bcnc$ on  lengths of paths with maximum weight
     containing no critical node.
For that, we first describe how to extract a simple path from an arbitrary
     path.

\subsection{Our first path reduction}

In this section, we construct from a path~$\pi$ its simple
     part~$\Simp(\pi)$ by repeatedly removing nonempty closed
     subpaths.
Each step of this construction corresponds to applying
Lemma~\ref{lem:path:split}. As the decomposition in this lemma is not unique,
we choose the to-be-removed closed subpath non-deterministically.
Formally, we fix a global choice
     function\footnote{We could also restrict the universe of possible nodes to a given set~$\mathcal{U}$ and explicitly state a choice function.} which we use every time we ``choose an $x$ in $X$''.

Let~$G$ be a graph and let~$\pi$ be a path in~$G$.
By Lemma~\ref{lem:path:split}, there exist
     paths~$\pi_1$, $\pi_2$ and a closed path~$\gamma$ such that $\pi=\pi_1\cdot \gamma\cdot
     \pi_2$.
Because $\End(\pi_1)=\Start(\pi_2)$, the concatenation $\pi_1\cdot\pi_2$
     is well-defined.
If~$\pi$ is non-simple, then we choose~$\gamma$ to be nonempty, i.e.,  $\ell(\pi_1\cdot\pi_2) < \ell(\pi)$.

We define~$\Step(\pi)$ to be the concatenation~$\pi_1\cdot\pi_2$.
If~$\pi$ is simple, then $\Step(\pi)=\pi$.
Furthermore, we define
     \[\Simp(\pi) = \lim_{t\to\infty}\Step^{t}(\pi)\enspace.\] 
The construction of $\Simp(\pi)$ takes a finite number of (at most~$\ell(\pi)$)
steps, hence $\Simp(\pi)$ is well-defined.
Since $\Step(\pi)=\pi$ if and only if~$\pi$ is simple, $\Simp(\pi)$
is simple.
Finally, $\pi$ and $\Step(\pi)$, and so $\pi$ and $\Simp(\pi)$, have the same start and end nodes,
respectively.
We call $\Simp(\pi)$ the {\em simple part\/} of~$\pi$.

\subsection{The critical bound}

\begin{lem}\label{lem:varrho:nc}
Let~$G$ be a nontrivial strongly connected en-weighted graph and
let~$\pi$ be a path in~$G$ whose nodes are non-critical.
Then, \[w(\pi) \leqslant w\big(\!\Simp(\pi)\big) + \varrho_{nc}(G) \cdot
     \bigr(\ell(\pi) - \ell\big(\!\Simp(\pi)\big) \bigr)\enspace.\]
\end{lem}
\begin{proof}
It suffices to show the inequality with $\Step(\pi)$ instead of
     $\Simp(\pi)$.

If $\Step(\pi) = \pi$, then the
     inequality trivially holds.
Otherwise, let  $\gamma$ be the nonempty closed path in the definition of
$\Step(\pi)$.
Then $w(\pi) = w\big(\!\Step(\pi)\big) + \wstar(\gamma)$.
By assumption,~$\gamma$ is a nonempty closed path whose nodes are non-critical, hence
     $\wstar(\gamma )\leqslant
     \varrho_{nc}(G)\cdot\ell(\gamma )$.
Noting $\ell(\gamma ) = \ell(\pi) - \ell\big(\! \Step(\pi)
     \big)$ concludes the proof.
\end{proof}

We introduce some additional  notation for a en-weighted graph~$G$: Let $\CDD_\nc(G)$ be the length
     of the longest simple path in~$G$ whose nodes are noncritical, and let
     $\CF_\cc(G)$ be the length of the longest elementary critical closed
     path in~$G$.
We define
\[\lVert w_V \rVert = \max_{i\in V}w_V(i) - \min_{i\in V}w_V(i)
     \enspace.\] 
Analogously to~$\lVert w_V\rVert$, we define for vectors~$v\in\IRmax^N$:
\[ \lVert v\rVert = \max_{1\leqslant i\leqslant N} v_i - \min_{1\leqslant i\leqslant N} v_i \]
To enhance readability and since no confusion can arise, we omit the
     dependency on the graph~$G$ in the next definition:
\begin{equation}\label{eq:n:zero}
\Bcnc(G) = \min\Bigg\{ 
\CDD_\nc   + \frac{\lVert w_V\rVert + \overline{\Delta}_{\nc}\,\CDD_{\nc}-\overline{\delta}\,\CDD}{-\overline{\varrho}_{\nc}}\ ,\ 
\frac{\lVert w_V\rVert + \overline{\Delta}_{\nc}\,\CDD_{\nc}-\overline{\delta}\,(N_{\nc}+\CF_\cc-1)}{-\overline{\varrho}_{\nc}}
\Bigg\}\enspace,
\end{equation}
where $N_\nc$ is the number of non-critical nodes of~$G$.


\begin{thm}\label{thm:n:zero}
Let~$G$ be a nontrivial strongly connected en-weighted graph and
     let~$i$ be a node of~$G$.
For all $n\geqslant \Bcnc(G)$, there exists a path of maximum
     en-weight in $\Pa^n(\ito,G)$ that contains a critical node.
\end{thm}
\begin{proof}
Since a path is of maximum en-weight in $\Pa^n(\ito,\overline{G})$ if
     and only if it is of maximum en-weight in
     $\Pa^n(\ito,{G})$, the critical nodes in $\overline{G}$
     and $G$ are the same, and $\Bcnc(\overline{G}) = \Bcnc(G)$, we
     may assume without loss of generality.\ that $\varrho(G) = 0$ in the following.

Let $\overline{\Delta}_\nc = \overline{\Delta}_\nc(G)$,
     $\overline{\delta} = \overline{\delta}(G)$, and
     $\overline{\varrho}_\nc = \overline{\varrho}_\nc(G)$.

Now suppose by contradiction that there exists an $n \ge \Bcnc(G)$
     such that all paths of maximum weight in
     $\Pa^n(\ito,G)$ are paths with non-critical nodes
     only.
Let~$\hat{\pi}$ be a path in~$\Pa^n(\ito,G)$ of maximum
     weight with non-critical nodes only.

Next choose a critical node~$k$ and a prefix~$\pi_c$
     of~$\Simp(\hat{\pi})$, such that the distance between~$k$ and
     $\End(\pi_c)$ is minimal.
Let~$\pi_2$ be a path of minimum length from~$\End(\pi_c)$ to~$k$.
Further let~$\gamma$ be a critical elementary closed path with
     $\Start(\gamma)=\End(\gamma)=k$.
Choose~$m \in \IN$ to be maximal such that $\ell(\pi_c)+\ell(\pi_2)+
     m\cdot\ell(\gamma) \leqslant n$ and choose~$\pi_1$ to be a prefix
     of $\gamma$ of length $n - \big(\ell(\pi_c)+\ell(\pi_2)+m\cdot
     \ell(\gamma)\big)$.
Clearly $\Start(\pi_1)=k$.
If we set $\pi = \pi_c\cdot\pi_2\cdot\gamma^m\cdot\pi_1$, we get
     $\ell(\pi)=n$ and for the weight of $\pi$ in~$G$,
\begin{equation}\label{eq:crit:lower:bound}
w(\pi) \geqslant \min_{j\in V}w_V(j) + w_*(\pi_c)+w_*(\pi_2)+w_*(\pi_1)\enspace.
\end{equation}
Figure~\ref{fig:thm:n:zero} illustrates path~$\pi$.

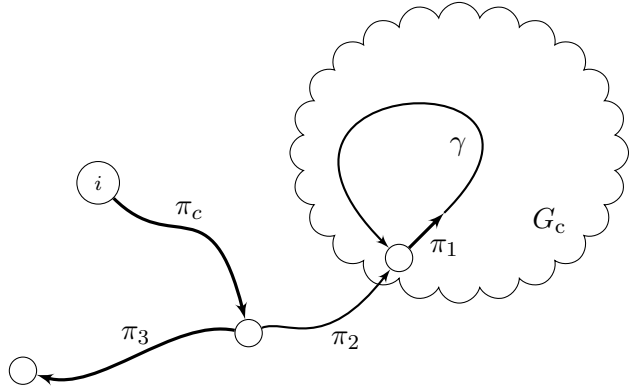
\begin{figure}[ht]
\centering
\begin{tikzpicture}[>=latex']
	\node[shape=circle,draw] (i) at (-2,2) {$\scriptstyle i$};
	\node[shape=circle,draw] (k) at (0,0) {};
	\node[shape=circle,draw] (j) at (2,1) {};
	\node[shape=circle,draw] (end) at (-3,-.5) {};
	\draw[very thick,->] (i) .. controls (-1,1) and (-.5,2)  .. node[midway,above]{${\pi}_c$}  (k);
	\draw[thick,->] (k) .. controls (0.5,0.2) and (1,-0.3)  .. node[near end,below=1mm]{${\pi}_2$}  (j);
	\draw[very thick,->] (j) --  (2.6,1.6) node[below=2mm] {$\pi_1$};
	\draw[thick,->] (2.6,1.6) .. controls +(2.1,2.1) and +(-2.5,2.5)  .. node[near start,below left]{$\gamma$}  (j);
	\draw[very thick,->] (k) .. controls +(-1,0.2) and +(1,-0.3) .. node[above]{$\pi_3$} (end);
	\node[cloud, cloud puffs=24, draw,minimum width=4.5cm, minimum height=4cm] at (2.8,2.4) {};
	\node at (4,1.5) {$G_\cc$};
\end{tikzpicture}
\caption{Path~$\pi$ in proof of Theorem~\ref{thm:n:zero}}
\label{fig:thm:n:zero}
\end{figure}

Let $\pi_3$ be a path such that $\Simp(\hat{\pi}) = \pi_c \cdot
     \pi_3$. By Lemma~\ref{lem:varrho:nc} we obtain for the weight of
     $\hat{\pi}$ in~$G$,
\begin{align}
w(\hat{\pi}) &\leqslant w\big(\!\Simp(\hat{\pi})\big) +
        \overline{\varrho}_{nc} \cdot \bigr(\ell(\hat{\pi}) - \ell\big(\!\Simp(\hat{\pi})\big) \bigr)\notag\\
 &\leqslant \max_{j\in V}w_V(j) + w_*(\pi_c)+w_*(\pi_3)+
        \overline{\varrho}_{nc}\cdot \bigr(\ell(\hat{\pi}) - \ell(\pi_c)-\ell(\pi_3) \bigr)\label{eq:upper:bound}
\end{align}
By assumption $w(\hat{\pi}) > w(\pi)$, and from
     \eqref{eq:crit:lower:bound}, \eqref{eq:upper:bound}, and $\overline{\varrho}_{nc} < 0$ we
     therefore obtain
\begin{align}
\ell(\hat{\pi}) &< \frac{\lVert w_V \rVert + w_*(\pi_3)-w_*(\pi_1)-w_*(\pi_2)}{-\overline{\varrho}_{nc}}+
                 \ell(\pi_3)+\ell(\pi_c)\notag\\
 &\le \frac{\lVert w_V \rVert + \overline{\Delta}_{nc}\,\ell(\pi_3)-\overline{\delta}\,(\ell(\pi_1)+\ell(\pi_2))}{-\overline{\varrho}_{nc}}+
                 \ell(\pi_3)+\ell(\pi_c)\label{eq:the:bound}
\end{align}
From \eqref{eq:the:bound} we may deduce,
\begin{equation}
\ell(\hat{\pi}) < \frac{\lVert w_V\rVert + \overline{\Delta}_{nc}\,\CDD_{nc}(G)-\overline{\delta}\,\CDD(G)}
                       {-\overline{\varrho}_{nc}}+\CDD_{nc}(G)\enspace.\label{eq:bound:1}
\end{equation}

Alternatively we may deduce from \eqref{eq:the:bound} with $-\overline{\varrho}_{nc}
     \le -\overline{\delta}$, $\ell(\pi_2)+\ell(\pi_3)+\ell(\pi_c) \le N_{nc}$,
     and $\ell(\pi_1) \le \CF_c(G)-1$ that        
\begin{align}
\ell(\hat{\pi}) &< \frac{\lVert w_V\rVert + \overline{\Delta}_{nc}\,\CDD_{nc}(G)-
   \overline{\delta}\,(N_{nc}+\CF_c(G)-1)}{-\overline{\varrho}_{nc}}\enspace.\label{eq:bound:2}
\end{align}
Combination of \eqref{eq:bound:1} and \eqref{eq:bound:2} yields a
     contradiction to $n \ge \Bcnc(G)$.
The lemma follows.
%
\end{proof}

Even and Rajsbaum~\cite[Lemma~10]{even:rajs} and Hartmann and
     Arguelles~\cite[Claim in proof of Theorem~10]{hartmann:arguelles}
     arrived at analog bounds for critical nodes on maximum weight
     paths.
A comparison of these bounds is given in Section~\ref{subsec:relation:work}

From Theorem~\ref{thm:n:zero} together with 
     $\CDD_\nc(G) \le N-1$ and $N_{\nc}+\CF_\cc(G) \le N$, and
     Lemma~\ref{lem:invariance} we immediately obtain:   
\begin{cor}\label{cor:n:zero}
Let~$G$ be a nontrivial strongly connected en-weighted graph with~$N$ nodes and let~$i$ be a node of~$G$.
For all
\begin{equation}
n\ge \frac{\lVert w_V\rVert + (\Delta_{\nc}(G)-\delta(G))\,(N-1)}{\varrho(G)-\varrho_\nc(G)} \ge \Bcnc(G)\enspace,\notag
\end{equation}
there exists a path of maximum en-weight in $\Pa^n(\ito,G)$ that
     contains a critical node.
\end{cor}

\section{Arbitrarily Long Closed Paths}\label{sec:explorationpenality}

We introduce for a strongly connected graph~$G$ the
    {\em exploration penalty\/} of~$G$,~$\EP(G)$, as the smallest integer $k$
	such that for any node $i$ and any integer~$n\geqslant k$ that is a 
	multiple of~$c(G)$, there is a closed path of length $n$ starting at~$i$.
We prove that $\EP(G)$ is finite, and we give an upper bound on $\EP(G)$ 
	which is  quadratic in the number of nodes of~$G$.
As we see in the subsequent sections, the exploration penalty plays a key role to bound 
	the transient as it constitutes a threshold to ``pump'' path weights  inside
	the critical graph.

\subsection{A number-theoretic lemma}\label{sec:numbertheorem}

We now state a useful number-theoretic lemma, which is a simple application
	of Brauer's Theorem~\cite{Bra42}.

Let ${\mathbf N}$ be any nonempty set of integers.
Any nonempty subset ${\mathbf A} \subseteq {\mathbf N}$ is said to be a {\em gcd-generator of\/ $\mathbf N$} if
	$\gcd({\mathbf A}) = \gcd({\mathbf N})$.
Note that, as $\mathds{Z}$ is Noetherian, any nonempty set of integers admits a finite gcd-generator.

\begin{lem}\label{lem:semigroup}
A set\/ $\mathbf N$ of positive integers that is closed under addition contains all but a finite
		number of multiples of its greatest common divisor.
Moreover, if\/ $\{a_1,\dots, a_k \}$ is a finite gcd-generator of\/~$\mathbf N$ with $a_1 \leq \dots \leq a_k$,
	then any multiple $n$ of $d=\gcd({\mathbf N})$ such that $n\geq d(\frac{a_1}{d}-1)(\frac{a_k}{d}-1)$
	is in\/ $\mathbf N$.
\end{lem}
\begin{proof}
Consider the set $\mathbf M$ of all the elements in $\mathbf N$ divided by $d =\gcd ({\mathbf N})$.
By Brauer's Theorem~\cite{Bra42}, we know that every integer $m \geq (\frac{a_1}{d}-1)(\frac{a_k}{d}-1)$
	is of the form $$ m = \sum_{i=1}^k x_i  \frac{a_i}{d}$$
	where each $x_i$ is a nonnegative integer.
Since $\mathbf N$ is closed under addition, it follows that  every  multiple of $d$ that is greater or equal 
	to $d(\frac{a_1}{d}-1)(\frac{a_k}{d}-1)$ is in $\mathbf N$.
In particular, 	all but a finite number of multiples of $d$ are in $\mathbf N$.
\end{proof}

\subsection{Constructing long paths}\label{sec:constructing}
In the case $G=(V,E)$ is a primitive graph, Denardo~\cite{denardo} established the following
	upper bound on~$\EP(G)$:

\begin{lem}[Denardo, {\cite[Corollary~1]{denardo}}]\label{lem:denardo}
Let $G$ be a strongly connected primitive graph with $N$ nodes and of girth $g$.
For any integer $n\geqslant N+(N-2)g$ and any node $i$ of $G$, there
     exists a closed path starting at~$i$ of length $n$.
\end{lem}

In~\cite{CBFN11b}, we prove that the same upper bound actually holds for any 
	(primitive or non-primitive) strongly connected graph.
We now show that we can obtain a better upper bound on $\EP(G)$ in the case of 
	non-primitive graphs from the number-theoretic lemma in Section~\ref{sec:numbertheorem}.
For that, we first recall some well-known properties on the lengths of paths in a 
	strongly connected graph.
Let ${\mathbf N}_{i,j}$ be the subset of integers defined by:
	$${\mathbf N}_{i,j} = \{ n\in \mathds{N} \mid \exists \pi \in \Pa(i,j,G), \  n = \ell(\pi)\}.$$
Clearly each ${\mathbf N}_{i,i}$ is closed under addition; let $d_i = \gcd ({\mathbf N}_{i,i})$.
Obviously, 
\begin{equation}\label{eqn:c}
	c(G) = \gcd (\{ d_i \mid i \in V \}).
\end{equation}

\begin{lem}\label{lem:Ni,j}
For any node $i$ in $G$, $d_i = c(G)$.
Moreover, for any pair of nodes $i,j$, all the elements in ${\mathbf N}_{i,j}$ have the same
	residue modulo $c(G)$.
\end{lem}

\begin{proof}
Let $i,j$ be any pair of nodes, and let $a \in {\mathbf N}_{i,j}$ and $b \in {\mathbf N}_{j,i}$.
The concatenation of a path from~$i$ to~$j$ with a path from~$j$ to~$i$ 
	is a closed path starting at~$i$.
Hence $a + b \in {\mathbf N}_{i,i}$. 
From Lemma~\ref{lem:semigroup}, we know that ${\mathbf N}_{j,j}$ contains all the multiples of $d_j$ 
	greater than some integer.
Consider any such multiple $k d_j$ with $k$ and $d_i$ relatively prime integers.
By inserting one corresponding closed path at node $j$ into the closed path at $i$
	with length $a+b$, we obtain a new closed path starting at $i$, i.e.,
	$a + k d_j + b \in {\mathbf N}_{i,i}$.
It follows that $d_i$ divides both $a+b$ and $a + k d_j + b$, and so~$d_i$ divides~$d_j$.
Similarly, we prove that~$d_j$ divides~$d_i$, and so $d_i = d_j$.
By (\ref{eqn:c}), the common value of the $d_i$'s  is actually equal to  $c(G)$.

Let $a$ and $a'$ be two  integers in ${\mathbf N}_{i,j}$.
The above argument gives both $a + b$ and $a' +b$ in ${\mathbf N}_{i,i}$.
Hence $c(G)$ divides $a + b $ and $a'+ b$, and so $a-a'$.
\end{proof}

\begin{lem}\label{lem:Ni,i:generator}
For any node $i$ of $G$, the set\/  ${\mathbf N}_{i,i}$ admits a gcd-generator 
	which contains all the lengths of  elementary closed paths starting at $i$, 
	and whose all elements $n$ satisfy the  inequality 
	$$ g \leq n \leq 2N -1$$
	where $g$ is the girth of~$G$ and~$N$ is the number of nodes in~$G$.
\end{lem}
\begin{proof}
Let $i$ be any node of $G$, and let $\gamma_0$ be any elementary closed path.
Let $\pi_1$ be one of the shortest paths from $i$ to $\gamma_0$,  and let $j= \End(\pi_1)$.
Without loss of generality, $\Start(\gamma_0)=j$.
By definition, $\ell(\pi_1) \leq N-\ell(\gamma_0)$.
Then consider a simple path $\pi_2$ from $j$ to $i$, and the two closed paths
	$$\pi = \pi_1 \cdot \pi_2 \mbox{ and } \pi' = \pi_1 \cdot \gamma_0 \cdot \pi_2\enspace.$$ 
Note that $$\ell(\pi) \leq \ell(\pi') \leq 2N-1$$
and $\ell(\pi)\geqslant g$, because~$\pi$ is closed.
In the particular case $i$ is a node of $\gamma_0$,  $\pi'$ reduces to $\gamma_0$,
	and so $\ell(\pi')$ is the length of the elementary closed path $\gamma_0$. 

Let $\mathbf{N}_i$ be the set of the lengths of the closed paths~$\pi$ and~$\pi'$ when 
	considering all the elementary closed paths~$\gamma_0$ in~$G$.
Then, $\mathbf{N}_i$ contains all the length of  elementary closed paths starting at $i$.
Let $g_i =\gcd(\mathbf{N}_i)$.
Since $\mathbf{N}_i \subseteq \mathbf{N}_{i,i}$,  $d_i$ divides $g_i$.
Conversely, let $\gamma_0$ be any elementary closed path, and let $\pi$ and $\pi'$
	be the two closed paths starting at node $i$ defined above;~$g_i$ divides both
	$\ell(\pi)$ and $\ell(\pi')$, and so divides $\ell(\pi') - \ell(\pi) = \ell(\gamma_0)$.
Hence,~$g_i$ divides the length of any elementary closed path, i.e., $g_i$ divides $c(G)$.
By Lemma~\ref{lem:Ni,j}, it follows that~$g_i$ divides~$d_i$.
Consequently, $g_i = d_i$, that is to say  $\mathbf{N}_i$ is a gcd-generator of $\mathbf{N}_{i,i}$.
\end{proof}

\begin{lem}\label{lem:bernadette:denardo}
Let $G$ be a strongly connected graph with $N$ nodes, of girth $g$ and cyclicity $c$.
For any node $i$ of $G$ and any integer 
	$n$ such that $n$ is a multiple of $c$  and 
	$n \geqslant 2Ng/c- g/c - 2g + c$,
	there exists a closed path of length~$n$ starting at~$i$. 
\end{lem}

\begin{proof}
Let $i$ be any node, and let $\gamma_0$ be any elementary closed path
	such that $\ell(\gamma_0) = g$.
Let~$\pi_1$ be one of the shortest paths from $i$ to $\gamma_0$,  and let $j= \End(\pi_1)$.
Without loss of generality, $\Start(\gamma_0)=j$.
By definition, $\ell(\pi_1) \leq N- g$.
Then consider an elementary path $\pi_2$ from $j$ to $i$; we have $\ell(\pi_2) \leq N-1$.
The path $\pi_1 \cdot \pi_2$ is closed at node $i$, and so $c$ divides 
	$\ell(\pi_1) + \ell(\pi_2)$.
Hence, if~$c$ divides some integer $n$, then $c$ also divides $n - \ell(\pi_1) - \ell(\pi_2)$.
It is $g\in \mathbf{N}_{j,j}$.
By Lemma~\ref{lem:Ni,i:generator}, there exists a gcd-generator~$\mathbf{N}_j$ of~$\mathbf{N}_{j,j}$ such that $g\in\mathbf{N}_j$ and $g\leqslant n\leqslant 2N-1$ for all $n\in\mathbf{N}_j$.

By Lemma~\ref{lem:semigroup}, for any~$n$ such that
	$n' = n - \ell(\pi_1) - \ell(\pi_2)$ is a multiple of~$c$ and
	$$ n' \geq c\,\left(\frac{g}{c} -1\right) \left(\frac{2N-1}{c} -1\right)\enspace,$$
	there exists a closed path $\gamma$ starting at node $j$ of length $\ell(\gamma) = n'$.
Note that
\[ c\left(\frac{g}{c}-1\right)\left(\frac{2N-1}{c}-1\right)+(N-g)+(N-1) = 2\frac{g}{c}N - \frac{g}{c} - 2g + c \enspace.\]
In this way, for any integer $n \geqslant 2Ng/c - g/c -2g+ c$ 
 	that is a multiple of $c$, we construct 
	$\pi = \pi_1 \cdot \gamma \cdot \pi_2$ that is a closed path at node $i$ of length $n$.
\end{proof}

We easily check that the upper bound on $\EP(G)$ given by Lemma~\ref{lem:bernadette:denardo}
 	is better than 
	the generalized Denardo bound~\cite{CBFN11b}, except when
        $c=1$, since then $2\le c\le g$ holds in the theorem below.

\begin{thm}\label{thm:EP}	
Let $G$ be a strongly connected graph with $N$ nodes, of girth $g$ and cyclicity $c$.
The exploration penalty of $G$, denoted $\EP$, is well-defined and satisfies the inequality 
	$$\EP \leq 2\frac{g}{c}N - \frac{g}{c} - 2g + c$$
	which can be improved to
	$$\EP \leq N+(N-2)g$$
	in case $G$ is a primitive graph.
\end{thm}

\section{Explorative Bound}\label{sec:explorative}

In this section, we show a bound on the transient~$n_{A,v}$ of linear max-plus systems for irreducible matrices~$A$ and vectors~$v$ with only finite entries, i.e., $v\in\IR^N$.
We construct arbitrarily long paths by {\em exploring\/} a critical
     component~$H$ in the sense of Section~\ref{sec:explorationpenality}; we hence call the resulting bound the {\em explorative bound}.
As we did in Section~\ref{sec:explorationpenality}, we distinguish the cases of whether the
     critical subgraph is primitive or not:
If~$c(H)=1$, we can find critical closed paths in~$H$ of arbitrary length~$t\geqslant \EP(H)$;
otherwise, it is only possible to find critical closed paths of length~$t$ if~$c(H)$ divides~$t$.
We start with the case of a primitive critical subgraph in Section~\ref{subsec:expl:prim}, because it allows for a simpler proof, as we do not have to consider the constraint that critical closed path lengths in critical component are necessarily multiples of~$c(H)$.

\subsection{Outline of the proof}\label{subsec:roadmap}

We now show how to establish an upper bound~$B(A,v)$ on the transient of the
sequence~$(A^{\otimes n}\otimes v)_i$
having the properties that (i) $B(A,v)$ is greater or equal to the critical bound~$B_\cc\big(G(A,v)\big)$ and (ii) $B(A,v)$ is invariant under substituting~$A$ by~$\overline{A}$.
By Lemmata~\ref{lem:hom} and~\ref{lem:transient:invariance}, and property~(ii), it suffices to
establish the upper bound in the case~$\varrho(A)=0$. 
Let $B=B(A,v)$.

First, we take any $\realrem{B}{n,p}$-realizer~$\pi$ for node~$i$ in en-weighted graph~$G=G(A,v)$; realizer~$\pi$  exists by Lemma~\ref{lem:realizers:exist}.
By property~(i) and Theorem~\ref{thm:n:zero}, we know that we can choose~$\pi$ to contain a critical node~$k$.
Depending on~$k$, we choose a modulus~$d$ dividing~$p(G)$.
In Section~\ref{subsec:expl:red}, we introduce a path reduction---a generalization of the simple part of a path---that preserves the residue class of the reduced path, i.e.,~$\ell(\pi)\equiv \ell(\hat{\pi}) \pmod{d}$ where~$\hat{\pi}$ is the reduced path of~$\pi$.
We apply our path reduction~$\Red_{d,k}$ to~$\pi$, obtaining path~$\hat{\pi}$.
Lemma~\ref{lem:red:upper:bound} in Section~\ref{subsec:expl:red} shows that the length of~$\hat{\pi}$ is less or equal to some bound~$B_{\Red}(d)$.
Further,~$k$ is a node of~$\hat{\pi}$ and $w(\hat{\pi})\geqslant w(\pi)$.

We then show that, for arbitrary multiples~$t$ of~$d$ such that $t\geqslant B-B_{\Red}(d)$, there exist critical closed paths starting at~$k$ of length~$t$.
It follows that for all~$n\geqslant B$, there exists an $\realrem{B}{n,p}$-realizer~$\pi_n$ for~$i$ of length~$n$.
Application of Lemma~\ref{lem:final:step} then concludes the proof.


\subsection{The case of primitive critical subgraphs}\label{subsec:expl:prim}


For a nontrivial strongly connected en-weighted graph~$G$ with primitive critical subgraph, define
\[ \Bep(G) = \max \biggr\{ \Bcnc(G) \ ,\ 2\cdot\CDD(G) +  \max_{H\in \SCC(G_c)} \EP(H) \biggr\} \enspace.\]


\begin{lem}\label{lem:realizer}
Let~$G$ be a nontrivial strongly connected en-weighted graph with primitive critical subgraph and~$\varrho(G)=0$; let~$B=\Bep(G)$ and~$p=p(G)$.
For all nodes~$i$ and all $n\geqslant B$ exists an $\realrem{B}{n,p}$-realizer for~$i$ of length~$n$.
\end{lem}
\begin{proof}
By Lemma~\ref{lem:realizers:exist}, there exists an $\realrem{B}{n,p}$-realizer~$\pi$ for~$i$.
By Theorem~\ref{thm:n:zero}, we can choose~$\pi$ to contain a critical node~$k$, because $B\geqslant \Bcnc(G)$.
Let~$\pi=\pi_1\cdot\pi_2$ with~$\End(\pi_1)=k$.

Set $\hat{\pi}_1=\Simp(\pi_1)$ and $\hat{\pi}_2=\Simp(\pi_2)$.
Because~$\varrho(G)=0$, it is $\varrho_\nc(G)\leqslant 0$, and hence $\wstar(\hat{\pi}_1)\geqslant \wstar(\pi_1)$ and
     $w(\hat{\pi}_2)\geqslant w(\pi_2)$ by Lemma~\ref{lem:varrho:nc}.
The lengths of~$\hat{\pi}_1$ and~$\hat{\pi}_2$ satisfy $\ell(\hat{\pi}_1) \leqslant \CDD(G)$ and
     $\ell(\hat{\pi}_2) \leqslant \CDD(G)$, because the paths are simple.

Let~$H\in\SCC(G_\cc)$ be the strongly connected component of~$G_c$ in which critical node~$k$ is contained.
Note that~$H$ is primitive, i.e., $c(H)=1$.
Thus, for all $t\geqslant \EP(H)$ there exists a critical closed
     path~$\gamma_t$ of length~$t$ starting at node~$k$.

By Lemma~\ref{lem:paths:in:crit:comps:are:critical}, $\wstar(\gamma_t)=\varrho(G)=0$, and hence
\begin{equation}\label{eq:realizer:blab}
w(\hat{\pi}_1 \cdot \gamma_t \cdot \hat{\pi}_2) = \wstar(\hat{\pi}_1) + \wstar(\gamma_t) + w(\hat{\pi}_2) \geqslant \wstar(\pi_1) + w(\pi_2) = w(\pi)\enspace.
\end{equation}
Further, $\ell(\hat{\pi}_1 \cdot \gamma_t \cdot \hat{\pi}_2) =
     \ell(\hat{\pi}_1) + \ell(\hat{\pi}_2) + t$.

We set $t=n-\big(\ell(\hat{\pi}_1) + \ell(\hat{\pi}_2)\big)
     \geqslant B-2\cdot\CDD(G) \geqslant \EP(H)$ and $\pi_n =
     \hat{\pi}_1 \cdot \gamma_t \cdot \hat{\pi}_2$.
Figure~\ref{fig:lem:realizer} depicts path~$\pi_n$.
It follows that $\ell(\pi_n)=n$ and $w(\pi_n)\geqslant w(\pi)$.
Hence, because~$n\geqslant B$ and~$\Start(\pi_n)=\Start(\pi)=i$, $\pi_n$ is an $\realrem{B}{n,p}$-realizer for~$i$ of length~$n$.
\begin{figure}[ht]
\centering
\begin{tikzpicture}[>=latex']
	\node[shape=circle,draw] (i) at (-4,-1) {$\scriptstyle i$};
	\node[shape=circle,draw] (k) at (0,0) {$\scriptstyle k$};
	\node[shape=circle,draw] (j) at (4,-1.2) {};
	\draw[thick,->] (i) .. controls (-3,1) and (-2,-2)  .. node[midway,above]{$\hat{\pi}_1$}  (k);
	\draw[thick,->] (k) .. controls (2,0.2) and (3,-2)  .. node[midway,above]{$\hat{\pi}_2$}  (j);
	\draw[thick] (k) .. controls (2,1.2) and (-1,2)  .. node[near start,left]{$\gamma_t$}  (1,2);
	\draw[thick,->] (1,2) .. controls (3,2.2) and (-1,1)  ..  (k);
	\node[cloud, cloud puffs=16, draw,minimum width=3.5cm, minimum height=4cm] at (0.6,1) {};
	\node at (1.8,1.5) {$H$};
\end{tikzpicture}
\caption{Realizer~$\pi_n$ in proof of Lemma~\ref{lem:realizer}}
\label{fig:lem:realizer}
\end{figure}
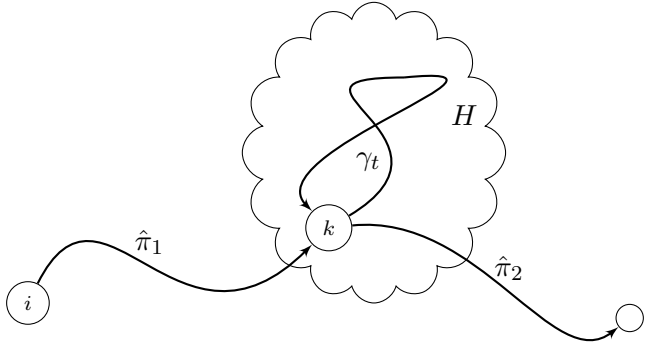
\end{proof}

\begin{thm}\label{thm:explorative:primitive}
For all irreducible $A\in\M_{N,N}(\IRmax)$ such that~$G_c(A)$ is primitive, and all~$v\in\IR^N$, we have 
$ n_{A,v} \leqslant \Bep\big( G(A,v) \big)$.
\end{thm}
\begin{proof}
By Lemma~\ref{lem:transient:invariance} and Lemma~\ref{lem:invariance}, both~$n_{A,v}$ and~$\Bep\big(G(A,v)\big)$ are invariant under substituting~$A$ by~$\overline{A}$. 
We may hence assume without loss of generality that~$\varrho(A)=0$, i.e., $\varrho\big(G(A,v)\big)=0$.

Let~$B=\Bep\big(G(A,v)\big)$.
Graph~$G(A,v)$ is nontrivial and strongly connected, because~$A$ is irreducible.
By Lemma~\ref{lem:realizer}, for every node~$i$ and every~$n\in\realrem{B}{n,p}$, there exists an $\realrem{B}{n,p}$-realizer for~$i$ of length~$n$.
Lemma~\ref{lem:final:step} hence implies that the transient of sequence~$\big(w^n(\ito)\big)_n$ in graph~$G(A,v)$ is at most~$B$.
But, by Lemma~\ref{lem:path:formula}, the transient of~$\big(w^n(\ito)\big)_n$ in~$G(A,v)$ is equal to~$n_{A,v}$.
This concludes the proof.
\end{proof}

\begin{cor}\label{cor:explorative:primitive}
For all irreducible $A\in\M_{N,N}(\IRmax)$ such that~$G_c(A)$ is primitive, and all~$v\in\IR^N$, we have 
\[ n_{A,v} \leqslant \max \left\{ \frac{\lVert v \rVert + \big(\Delta_\nc(G) - \delta(G)\big) \cdot(N-1) }{ \varrho(G) - \varrho_\nc(G)}  \ ,\ 2\cdot(N-1) + \max_{H\in\SCC(G_c)} \EP(H)  \right\} \enspace,\]
where~$G=G(A)$.
\end{cor}

\subsection{Our second path reduction}\label{subsec:expl:red}

In the case that the critical subgraph is not primitive, i.e., if there exist critical components~$H$ with~$c(H)>1$, our first path reduction---the simple part of a path---is not sufficient for us, because we can no longer choose a critical path~$\gamma_t$ in~$H$ of arbitrary length~$t\geqslant \EP(H)$, but only of lengths that are multiples of~$c(H)$.
In this section, we thus introduce a generalization of our first path
     reduction~$\Simp(\pi)$, which is able to preserve the residue class modulo~$c(H)$ of the path lengths.
The idea is to repeatedly remove {\em collections\/} of closed subpaths of~$\pi$ whose combined length is a multiple of~$c(H)$.

The generalized path reduction, $\Red_{d,k}(\pi)$, has two additional
     parameters: a modulus~$d$ and a node~$k$ of~$\pi$.
The reduction is defined in such a way that (a) the reduced
     path's length is in the same residue class modulo~$d$ as
     path~$\pi$'s length, and (b) node~$k$ is a node of the reduced path.
Our first path reduction is a special case of this second path
     reduction when setting~$d=1$ and~$k$ to be the start node
     of~$\pi$.

Let~$G$ be a graph and let~$\pi$ be a path in~$G$.
Let~$\mathcal{S}$ be a finite multiset of subpaths of~$\pi$.
We say that~$\mathcal{S}$ is {\em disjoint\/} if there exist
     paths~$\sigma_0,\sigma_1,\dots,\sigma_n$ such that 
\begin{equation}\label{eq:def:disjoint}
\pi = \sigma_0\cdot\pi_1\cdot\sigma_1\cdots\pi_n\cdot\sigma_n
\end{equation}
where $\mathcal{S} = \{ \pi_1,\pi_2,\dots,\pi_n \}$.

For a disjoint multiset~$\mathcal{S}$ of {\em closed\/} subpaths of~$\pi$, define
     \[ \Rem(\pi,\mathcal{S}) = \sigma_0\cdot\sigma_1\cdots\sigma_n \] 
where the~$\sigma_l$ are chosen to fulfill Equation~\eqref{eq:def:disjoint}.
The paths~$\pi$ and~$\Rem(\pi,\mathcal{S})$ have the same start and end nodes, respectively.
Furthermore $\ell\big(\!\Rem(\pi,\mathcal{S})\big) = \ell(\pi) - L(\mathcal{S})$
      where $L(\mathcal{S}) = \sum_{\gamma\in\mathcal{S}}\ell(\gamma)$.
In particular,~$\Rem(\pi,\mathcal{S})=\pi$ if and only if~$L(\mathcal{S})=0$.
Hence if~$\Rem(\pi,\mathcal{S})\neq\pi$, then necessarily $\ell\big(\! \Rem(\pi,\mathcal{S}) \big) < \ell(\pi)$.

For a path~$\pi$ and a node~$k$ of~$\pi$, denote by~${\mathbf S}_k(\pi)$ the set of
     disjoint multisets~$\mathcal{S}$ of elementary closed subpaths of~$\pi$ such that~$k$ is a node of~$\Rem(\pi,\mathcal{S})$.     
For a path~$\pi$, a node~$k$ of~$\pi$, and a positive integer~$d$, define
\[{\mathbf S}_{d,k}(\pi) = \left\{ \mathcal{S} \in {\mathbf S}_k(\pi) \mid L(\mathcal{S}) \equiv 0 \pmod d \right\}\enspace.\] 
This set is never empty, because~$k$ is a node of~$\pi$ and we can hence choose~$\mathcal{S}$ to be empty.

Choose~$\mathcal{S}\in {\mathbf S}_{d,k}(\pi)$ such that~$L(\mathcal{S})$ is maximal.
Then set $\Step_{d,k}(\pi)=\Rem(\pi,\mathcal{S})$ and \[\Red_{d,k}(\pi) = \lim_{t\to\infty}
     \Step_{d,k}^t(\pi)\enspace.\] 
The construction of~$\Red_{d,k}(\pi)$ takes a finite number of (at most~$\ell(\pi)$) steps, hence~$\Red_{d,k}(\pi)$ is well-defined.
It is~$\Red_{d,k}(\pi) = \pi$ if and only if~$L(\mathcal{S})=0$ for all~$\mathcal{S}\in\mathbf{S}_{d,k}(\pi)$.
The paths $\pi$ and $\Red_{d,k}(\pi)$ have the same start and end nodes, respectively.
Also,~$k$ is a node of~$\Red_{d,k}(\pi)$ and
\[ \ell\big(\! \Red_{d,k}(\pi) \big) \equiv \ell(\pi) \pmod d \enspace. \]
Finally, whenever~$G$ is an en-weighted graph with~$\varrho(G)=0$, $w\big(\! \Rem(\pi,\mathcal{S}) \big) \geqslant w(\pi)$, because we repeatedly remove closed subpaths.

The following lemma is a well-known elementary application of the pigeonhole principle.
Erd\H{o}s attributed it to V{\'a}zsonyi and Sved. \cite[p.~133]{proofs:from:the:book}

\begin{lem}\label{lem:erdos}
Let $d$ be a positive integer and let $x_1,\dots,x_d\in\mathds{Z}$.
Then there exists a nonempty $I\subseteq\{1,\dots,d\}$ such that  
\[ \sum_{i\in I} x_i \equiv 0 \pmod d\enspace. \]
\end{lem}
\begin{proof}
For $1\leqslant k\leqslant d$, set $S_k = \sum_{i=1}^k x_i$.
If there exist $k<\ell$ with $S_k\equiv S_\ell \pod d$, then set
     $I=\{k+1,k+2,\dots,\ell\}$.
Otherwise the mapping $k \mapsto (S_k \bmod d)$ is injective, hence
     surjective onto $\{0,1,\dots,d-1\}$, which implies that there
     exists a $k_0$ with $S_{k_0}\equiv 0 \pod d$.
In this case, set $I=\{1,2,\dots,k_0\}$.
\end{proof}

With the help of Lemma~\ref{lem:erdos}, we can prove the following upper bound on the length of the reduced path~$\Red_{d,k}(\pi)$.

\begin{lem}\label{lem:red:upper:bound}
Let~$G$ be a graph with~$N$ nodes.
For all positive integers~$d$, all nodes~$k$, and all paths~$\pi$ that contain node~$k$,
\[ \ell\big(\!\Red_{d,k}(\pi) \big)  \leqslant (d-1)\cdot \CF(G) + (d+1) \cdot \CDD(G) \enspace. \]
\end{lem}
\begin{proof}
Write $\hat{\pi} = \Red_{d,k}(\pi)$.
Let $\hat{\mathcal{S}}\in {\mathbf S}_{k}(\hat{\pi})$ such that $L(\hat{\mathcal{S}})$ is maximal.
Further, let~$\mathcal{S}$ be the sub-multiset of the {\em nonempty\/} closed paths of~$\hat{\mathcal{S}}$.
It is~$L(\mathcal{S}) = L(\hat{\mathcal{S}})$.
Because $\hat{\mathcal{S}}\in{\mathbf S}_k(\hat{\pi})$ and $\mathcal{S}\subseteq\hat{\mathcal{S}}$, it follows that $\mathcal{S}\in{\mathbf S}_k(\hat{\pi})$.

We first show that $\lvert \mathcal{S}\rvert \leqslant d-1$:
 Otherwise, by
     Lemma~\ref{lem:erdos}, there exists a
     nonempty~$\mathcal{S}'\subseteq\mathcal{S}$
     with $L(\mathcal{S}')=\sum_{\gamma\in\mathcal{S}'} \ell(\gamma)\equiv 0\pmod d$; hence~$\mathcal{S}'\in\mathbf{S}_{d,k}(\hat{\pi})$ with $L(\mathcal{S}')>0$, a contradiction to~$\Red_{d,k}(\hat{\pi})=\hat{\pi}$.
We have thus proved
\begin{equation}\label{eq:t:upper:bound}
\lvert \mathcal{S}\rvert \leqslant d-1\enspace.
\end{equation}
Every~$\gamma\in \mathcal{S}$ is elementary, therefore~$\ell(\gamma)\leqslant \CF(G)$ and
     thus 
\begin{equation}\label{eq:ls:upper}
L(\mathcal{S}) \leqslant (d-1)\cdot \CF(G)\enspace.
\end{equation}

Let $n=\lvert\mathcal{S}\rvert$, $\mathcal{S} = \{\gamma_1,\dots,\gamma_n\}$, and 
\[ \hat{\pi} = \sigma_0 \cdot \gamma_1 \cdot \sigma_1 \cdots \gamma_n \cdot \sigma_n \enspace. \]
Because~$k$ is a node of~$\Rem(\hat{\pi},\mathcal{S})$, there exists an~$r$ such that~$k$ is a node of~$\sigma_r$.
Figure~\ref{fig:lem:red:upper:bound} shows this decomposition of path~$\hat{\pi}$.
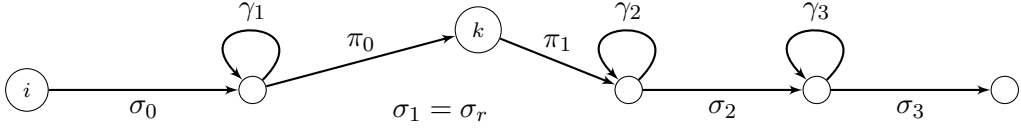
\begin{figure}[ht]
\centering
\begin{tikzpicture}[>=latex']
	\node[shape=circle,draw] (i) at (-6,0) {$\scriptstyle i$};
	\node[shape=circle,draw] (k) at (0,.8) {$\scriptstyle k$};
	\node[shape=circle,draw] (j) at (7,0) {};
	\node[shape=circle,draw] (n1) at (4.5,0) {};
	\node[shape=circle,draw] (n2) at (2,0) {};
	\node[shape=circle,draw] (n3) at (-3,0) {};
	\node  at (-0.5,-0.3) {$\sigma_1=\sigma_r$};
	\draw[thick,->] (i) -- node[below] {$\sigma_0$} (n3);
	\draw[thick,->] (n3) -- node[above] {$\pi_0$} (k);
	\draw[thick,->] (k) -- node[above] {$\pi_1$} (n2);
	\draw[thick,->] (n2) -- node[below] {$\sigma_2$} (n1);
	\draw[thick,->] (n1) -- node[below] {$\sigma_3$} (j);
	\draw[thick,->] (n3) .. controls +(1,1) and +(-1,1) .. node[above] {$\gamma_1$} (n3);
	\draw[thick,->] (n2) .. controls +(1,1) and +(-1,1) .. node[above] {$\gamma_2$} (n2);
	\draw[thick,->] (n1) .. controls +(1,1) and +(-1,1) .. node[above] {$\gamma_3$} (n1);
\end{tikzpicture}
\caption{Path~$\hat{\pi}$ in proof of Lemma~\ref{lem:red:upper:bound}}
\label{fig:lem:red:upper:bound}
\end{figure}

If~$m\neq r$, then we show that~$\sigma_m$ is simple:
Otherwise by Lemma~\ref{lem:path:split}, there exists a nonempty elementary closed subpath~$\gamma'$ of $\sigma_m$.
But then $\mathcal{S}'=\mathcal{S}\cup\{\gamma'\}\in {\mathbf S}_{k}(\hat{\pi})$, because~$k$ is a node of~$\sigma_r$ and~$m\neq r$; a contradiction to maximality of~$L(\mathcal{S})$, because $L(\mathcal{S}') > L(\mathcal{S})$.
Hence 
\begin{equation}\label{eq:node:disj:upper:bound}
\ell(\sigma_m)\leqslant \CDD(G) \ \text{ for }\ m\neq r\enspace.
\end{equation}

We now show that $\ell(\sigma_r)\leqslant 2\cdot \CDD(G)$:
Let~$\sigma_r=\pi_0\cdot\pi_1$ with $\End(\pi_0)=k$.
Then~$\pi_0$ and~$\pi_1$ are simple, because otherwise, if~$\gamma_t'$ is a nonempty elementary closed subpath of~$\pi_t$,~$t\in\{0,1\}$, then $\mathcal{S}'=\mathcal{S}\cup\{\gamma_t'\}\in{\mathbf S}_k(\hat{\pi})$, because~$k$ is a node of~$\pi_{1-t}$; a contradiction to the maximality of~$L(\mathcal{S})$, because  $L(\mathcal{S}') > L(\mathcal{S})$.
We have thus proved that  
\begin{equation}\label{eq:nnd:upper:bound}
\ell(\sigma_r)\leqslant 2\cdot\CDD(G)\enspace. 
\end{equation}

We note that $n=\lvert\mathcal{S}\rvert\leqslant d-1$ by~\eqref{eq:t:upper:bound}. Thus combination of 
     \eqref{eq:node:disj:upper:bound} and~\eqref{eq:nnd:upper:bound}
     yields
\begin{equation}
\sum_{m=0}^n \ell(\sigma_m) \leqslant (d+1)\cdot \CDD(G) \enspace,
\end{equation}
which, in turn, yields by combination with~\eqref{eq:ls:upper}:
\begin{equation*}
\ell(\hat{\pi}) = L(\mathcal{S}) + \sum_{m=0}^n \ell(\sigma_m) \leqslant (d-1)\cdot\CF(G) + (d+1)\cdot\CDD(G)
\end{equation*}
This concludes the proof.
\end{proof}

\subsection{Extension to the general case}\label{subsec:expl:nonprim}

In this section, we generalize the explorative bound to the case of a not necessarily primitive critical subgraph.
For that, we use the generalized path reduction~$\Red_{d,k}(\pi)$.

For a nontrivial strongly connected en-weighted graph~$G$ define
\[ \Benp(G) = \max \left\{ \Bcnc(G)\ ,\ \big(d(G_c)-1\big)\cdot \CF(G) + \big(d(G_c)+1\big)\cdot \CDD(G) + \max_{H\in\SCC(G_c)} \EP(H) \right\} \enspace. \]
This definition generalizes the previous definition of~$\Bep(G)$, which was stated for the case of a primitive critical subgraph.

\begin{lem}\label{lem:realizer:gen}
Let~$G$ be a nontrivial strongly connected en-weighted graph with~$\varrho(G)=0$;
let~$B=\Benp(G)$ and~$p=p(G)$.
For all nodes~$i$ and all $n\geqslant B$ exists an $\realrem{B}{n,p}$-realizer for~$i$ of length~$n$.
\end{lem}
\begin{proof}
By Lemma~\ref{lem:realizers:exist}, there exists an $\realrem{B}{n,p}$-realizer~$\pi$ for~$i$, i.e., $\pi\in P(\ito)$ is of maximum weight such that
     $\ell(\pi)\geqslant B$ and $\ell(\pi) \equiv n \pmod {p} $.
By Theorem~\ref{thm:n:zero}, we can choose~$\pi$ such that there is a critical node~$k$
     on~$\pi$, because~$B\geqslant \Bcnc(G)$.
Let~$H\in\SCC(G_c)$ be the critical component of~$k$.

Set $\hat{\pi} = \Red_{c(H),k}(\pi)$.
It is $\ell(\hat{\pi}) \equiv \ell(\pi) \pod {c(H)}$ and also $\ell(\pi)\equiv n \pod{c(H)}$, because~$c(H)$ divides~$p(G)$.
Hence $\ell(\hat{\pi})\equiv n \pod{c(H)}$, i.e., there exists an
     integer~$m$ such that $n = \ell(\hat{\pi}) + m\cdot
     c(H)$.
It follows from Lemma~\ref{lem:red:upper:bound} and $c(H)\leqslant d(G_c)$ that $m\cdot c(H)
     = n - \ell(\hat{\pi}) \geqslant \EP(H)$; in particular,~$m$ is nonnegative.
Hence, by Lemma~\ref{lem:bernadette:denardo}, there exists a closed
     path~$\gamma_t$ of length~$t=m\cdot c(H)$ in~$H$ starting at
     node~$k$.
By Lemma~\ref{lem:paths:in:crit:comps:are:critical}, $\wstar(\gamma_t)=\varrho(G)=0$.

Let~$\hat{\pi}=\hat{\pi}_1\cdot\hat{\pi}_2$ with~$\End(\hat{\pi}_1)=k$ and set~$\pi_n =
     \hat{\pi}_1\cdot \gamma_t\cdot \hat{\pi}_2$.
Path~$\pi_n$ has the same form as in Figure~\ref{fig:lem:realizer}.
Then $\ell(\pi_n) = \ell(\hat{\pi}) + \ell(\gamma_t) = n$
and \[w(\pi_n) = \wstar(\hat{\pi}_1) + \wstar(\gamma_t) + w(\hat{\pi}_2) = w(\hat{\pi}_1\cdot\hat{\pi}_2) = w(\hat{\pi}) \geqslant w(\pi)\enspace.\]
Hence, because~$\Start(\pi_n)=\Start(\pi)=i$, $\pi_n$ is an $\realrem{B}{n,p}$-realizer for~$i$ of length~$n$.
\end{proof}

\begin{thm}\label{thm:explorative:nonprimitive}
For all irreducible $A\in\M_{N,N}(\IRmax)$ and all~$v\in\IR^N$, we have 
$ n_{A,v} \leqslant \Benp\big( G(A,v) \big)$.
\end{thm}
\begin{proof}
By Lemma~\ref{lem:transient:invariance} and Lemma~\ref{lem:invariance}, both~$n_{A,v}$ and~$\Bep\big(G(A,v)\big)$ are invariant under substituting~$A$ by~$\overline{A}$. 
We may hence assume without loss of generality that~$\varrho(A)=0$, i.e., $\varrho\big(G(A,v)\big)=0$.

Let~$B=\Benp\big(G(A,v)\big)$ and~$p=p(G)$.
Graph~$G(A,v)$ is nontrivial and strongly connected, because~$A$ is irreducible.
By Lemma~\ref{lem:realizer:gen}, for every node~$i$ and every~$n\geqslant B$, there exists an $\realrem{B}{n,p}$-realizer for~$i$ of length~$n$.
Lemma~\ref{lem:final:step} hence implies that the transient of sequence~$\big(w^n(\ito)\big)_n$ in graph~$G(A,v)$ is at most~$B$.
But, by Lemma~\ref{lem:path:formula}, the transient of~$\big(w^n(\ito)\big)_n$ in~$G(A,v)$ is equal to~$n_{A,v}$.
This concludes the proof.
\end{proof}

\begin{cor}\label{cor:explorative:nonprimitive}
For all irreducible $A\in\M_{N,N}(\IRmax)$ and all~$v\in\IR^N$, we have 
\[ n_{A,v} \leqslant \max \left\{ \frac{\lVert v \rVert + \big(\Delta_\nc(G) - \delta(G)\big) \cdot(N-1) }{ \varrho(G) - \varrho_\nc(G)}  \ ,\ (d-1) + 2d\cdot(N-1) + \max_{H\in\SCC(G_c)} \EP(H)  \right\} \enspace,\]
where~$G=G(A)$ and $d=d(G_c)$.
\end{cor}

\section{Repetitive Bounds}\label{sec:nonexplorative}

In Section~\ref{sec:explorative}, we explored critical
     components~$H$ in the sense
     of Section~\ref{sec:explorationpenality} to construct critical closed paths
     whose lengths are multiples of~$c(H)$.
We now follow another approach, which yields better
     results in some cases: We do not explore anymore
     critical components, but instead consider a single (elementary) critical
     closed path~$\gamma$  starting at  critical node~$k$,
     and repeatedly add~$\gamma$ to the constructed paths; we hence call the resulting bounds {\em repetitive}.
We present two repetitive bounds: one using our path reduction~$\Red_{d,k}$ in Section~\ref{subsec:our:repetitive}, and an improvement of the bound of Hartmann and Arguelles~\cite{hartmann:arguelles} by a factor of two in Section~\ref{subsec:improve}.

\subsection{Our repetitive bound}\label{subsec:our:repetitive}
This section presents the bound we get when substituting the exploration of the visited critical component by only using a single elementary critical closed path in the method used to derive the explorative bound.
The proof follows the same outline as described in Section~\ref{subsec:roadmap}.
In the resulting upper bound on the transient, we
     substitute~$c(H)$ by~$\ell(\gamma)$, and thus substitute~$d(G_c)$ by~$\CF(G_c)$, and drop the exploration
     penalty terms~$\EP(H)$.
It is of course possible that~$\CF(G_c)$ is strictly greater than~$d(G_c)$.

For a nontrivial strongly connected en-weighted graph~$G$ define
\[ \Bneone(G) = \max \biggr\{ \Bcnc(G) \ ,\ \big( \CF(G_c) - 1 \big)\cdot \CF(G) + \big( \CF(G_c) + 1 \big)\cdot \CDD(G)    \biggr\} \enspace. \]

\begin{lem}\label{lem:realizer:non1}
Let~$G$ be a nontrivial strongly connected en-weighted graph with~$\varrho(G)=0$; let~$B=\Bneone(G)$ and~$p=p(G)$.
For all nodes~$i$ and all $n\geqslant B$ exists an $\realrem{B}{n,p}$-realizer for~$i$ of length~$n$.
\end{lem}
\begin{proof}
By Lemma~\ref{lem:realizers:exist}, there exists an $\realrem{B}{n,p}$-realizer~$\pi$ for~$i$, i.e., $\pi\in P(\ito)$ is of maximum weight such that
     $\ell(\pi)\geqslant B$ and $\ell(\pi) \equiv n \pmod {p} $.
By Theorem~\ref{thm:n:zero}, we can choose~$\pi$ such that there is a critical node~$k$
     on~$\pi$, because~$B\geqslant \Bcnc(G)$.
Let~$\gamma\in\SCC(G_c)$ be an elementary critical closed path starting at~$k$.

Set $\hat{\pi} = \Red_{\ell(\gamma),k}(\pi)$.
It is $\ell(\hat{\pi}) \equiv \ell(\pi) \pod {\ell(\gamma)}$ and also $\ell(\pi)\equiv n \pod{\ell(\gamma)}$, because~$\ell(\gamma)$ divides~$p(G)$.
Hence $\ell(\hat{\pi})\equiv n \pod{\ell(\gamma)}$, i.e., there exists an
     integer~$m$ such that $n = \ell(\hat{\pi}) + m\cdot
     \ell(\gamma)$.
Lemma~\ref{lem:red:upper:bound} and~$\ell(\gamma)\leqslant \CF(G_c)$ implies $\ell(\hat{\pi})\leqslant B$, hence~$m$ is nonnegative.

Let~$\hat{\pi}=\hat{\pi}_1\cdot\hat{\pi}_2$ with~$\End(\hat{\pi}_1)=k$ and set~$\pi_n =
     \hat{\pi}_1\cdot \gamma^m\cdot \hat{\pi}_2$.
Figure~\ref{fig:lem:realizer:non1} depicts path~$\pi_n$.
Then $\ell(\pi_n) = \ell(\hat{\pi}) + m\cdot\ell(\gamma) = n$.
Further, $w(\pi_n) = w(\hat{\pi}) \geqslant w(\pi)$ and $\Start(\pi_n)=\Start(\pi)=i$, i.e., $\pi_n$ is an $\realrem{B}{n,p}$-realizer for~$i$ if length~$n$.
\end{proof}
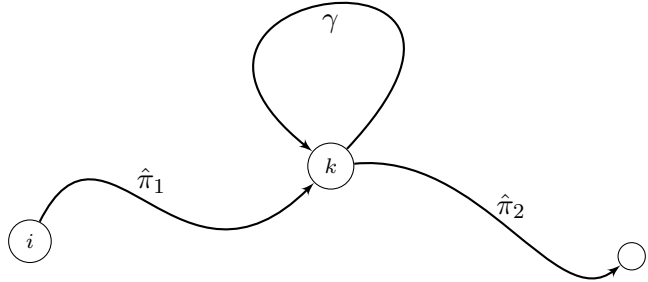
\begin{figure}[ht]
\centering
\begin{tikzpicture}[>=latex']
	\node[shape=circle,draw] (i) at (-4,-1) {$\scriptstyle i$};
	\node[shape=circle,draw] (k) at (0,0) {$\scriptstyle k$};
	\node[shape=circle,draw] (j) at (4,-1.2) {};
	\draw[thick,->] (i) .. controls (-3,1) and (-2,-2)  .. node[midway,above]{$\hat{\pi}_1$}  (k);
	\draw[thick,->] (k) .. controls (2,0.2) and (3,-2)  .. node[midway,above]{$\hat{\pi}_2$}  (j);
	\draw[thick,->] (k) .. controls +(3,3.2) and +(-3,2.4)  .. node[below] {$\gamma$}  (k);
\end{tikzpicture}
\caption{Realizer~$\pi_n$ in proof of Lemma~\ref{lem:realizer:non1}}
\label{fig:lem:realizer:non1}
\end{figure}

\begin{thm}\label{thm:nonexplorative}
For all irreducible $A\in\M_{N,N}(\IRmax)$ and all~$v\in\IR^N$, we have
$ n_{A,v} \leqslant \Bneone\big( G(A,v) \big)$.
\end{thm}
\begin{proof}
By Lemma~\ref{lem:transient:invariance} and Lemma~\ref{lem:invariance}, both~$n_{A,v}$ and~$\Bneone\big(G(A,v)\big)$ are invariant under substituting~$A$ by~$\overline{A}$. 
We may hence assume without loss of generality that~$\varrho(A)=0$, i.e., $\varrho\big(G(A,v)\big)=0$.

Let~$B=\Bneone\big(G(A,v)\big)$ and~$p=p(G)$.
Graph~$G(A,v)$ is nontrivial and strongly connected, because~$A$ is irreducible.
By Lemma~\ref{lem:realizer:non1}, for every node~$i$ and every~$n\geqslant B$, there exists an $\realrem{B}{n,p}$-realizer for~$i$ of length~$n$.
Lemma~\ref{lem:final:step} hence implies that the transient of sequence~$\big(w^n(\ito)\big)_n$ in graph~$G(A,v)$ is at most~$B$.
But, by Lemma~\ref{lem:path:formula}, the transient of~$\big(w^n(\ito)\big)_n$ in~$G(A,v)$ is equal to~$n_{A,v}$.
This concludes the proof.
\end{proof}

\begin{cor}\label{cor:nonexplorative}
For all irreducible $A\in\M_{N,N}(\IRmax)$ and all~$v\in\IR^N$, we have 
\[ n_{A,v} \leqslant \max \left\{ \frac{\lVert v \rVert + \big(\Delta_\nc(G) - \delta(G)\big) \cdot(N-1) }{ \varrho(G) - \varrho_\nc(G)}  \ ,\  (\CF-1) + 2\,\CF\cdot(N-1) \right\} \enspace\]
where~$G=G(A)$ and~$\CF=\CF(G_c)$.
\end{cor}

\subsection{Improved Hartmann-Arguelles bound}\label{subsec:improve}


In Corollary~\ref{cor:nonexplorative}, the second term in the maximum is bounded by~$2N^2$.
Hartmann and Arguelles also arrived at the term~$2N^2$ in the corresponding part of their upper bound on the transient of a max-plus system~\cite[Theorem~12]{hartmann:arguelles}.
To prove this upper bound, Hartmann and Arguelles used a different kind of path reduction, which we recall in Theorem~\ref{thm:ha:thm4} below.
This path reduction guarantees that the reduced path is in the same residue class modulo~$\ell(\gamma)$, the length of some visited critical closed path~$\gamma$.

Hartmann and Arguelles, after reducing a realizer with their path reduction, combine elementary critical closed paths in the visited critical component to arrive at the right residue class modulo~$c(G_c)$.
However, this construction is not necessary, for it is not necessary to consider~$c(G_c)$ as the period, as we have shown in Section~\ref{subsec:period}.
We can improve on their proof by considering~$p(G)$ as the period, which is a multiple of~$c(G_c)$.
Avoiding the construction saves one time~$N^2$ in the resulting bound; hence the second term in the maximum can chosen to be~$N^2$ instead of~$2N^2$, as we show now.

\begin{thm}[{\cite[Theorem~4]{hartmann:arguelles}}]\label{thm:ha:thm4}
Let~$G$ be an e-weighted graph with~$N$ nodes and~$\varrho(G)=0$.
Let~$\pi\in\Pa(i,j,G)$ with~$\ell(\pi)\geqslant N^2$ and let~$k$ be a
     critical node of~$\pi$.
Further, let~$\gamma$ be a critical closed path starting from~$k$.
Then there exist a path~$\hat{\pi}\in\Pa(i,j,G)$ such that $\ell(\hat{\pi})\equiv\ell(\pi)\pmod{\ell(\gamma)}$,
     $\wstar(\hat{\pi}) \geqslant \wstar(\pi)$,~$k$ is a node of~$\hat{\pi}$, and
     $\ell(\hat{\pi})<N^2$.
\qed    
\end{thm}

For a nontrivial strongly connected en-weighted graph~$G$ with $N$ nodes define
\[ \Bnetwo(G) = \max\big\{ \Bcnc(G)\,,\,N^2 \big\} \enspace. \]

\begin{lem}\label{lem:realizer:ha}
Let~$G$ be a nontrivial strongly connected en-weighted graph with~$\varrho(G)=0$; let~$B=\Bnetwo(G)$ and~$p=p(G)$.
For all nodes~$i$ and all $n\geqslant B$ exists an $\realrem{B}{n,p}$-realizer for~$i$ of length~$n$.
\end{lem}
\begin{proof}
By Lemma~\ref{lem:realizers:exist}, there exists a $\realrem{B}{n,p}$-realizer~$\pi$ for~$i$, i.e., $\pi\in\Pa(\ito)$ is of maximum weight such that
     $\ell(\pi)\geqslant B$ and $\ell(\pi) \equiv n \pmod{p}$.
By Theorem~\ref{thm:n:zero}, we can choose~$\pi$ such that there is a critical
     node~$k$ on~$\pi$, because~$B\geqslant\Bcnc$.
Let~$\gamma$ be an elementary critical closed path starting at~$k$.

Let~$\hat{\pi}$ be as in Theorem~\ref{thm:ha:thm4}
and let~$\hat{\pi}=\hat{\pi}_1\cdot\hat{\pi}_2$ with~$\End(\hat{\pi}_1)=k$.
It is $\ell(\hat{\pi}) \equiv \ell(\pi)\pmod
     {\ell(\gamma)}$ and also $\ell(\pi)\equiv r \pmod{\ell(\gamma)}$, because~$\ell(\gamma)$ divides~$p(G)$.
Hence $\ell(\hat{\pi})\equiv n\pmod{\ell(\gamma)}$, i.e., there exists an integer~$m$ such that
     $t=m\cdot \ell(\gamma) = n - \ell(\hat{\pi})$.
It is $\ell(\hat{\pi})< N^2\leqslant B$, hence~$m$ is nonnegative.

Setting~$\pi_n=\hat{\pi}_1\cdot\gamma^m\cdot\hat{\pi}_2$ yields~$\ell(\pi_n)=n$
     and~$w(\pi_n)=w(\hat{\pi})\geqslant w(\pi)$.
This concludes the proof.
\end{proof}

\begin{thm}\label{thm:nonexplorative2}
For all irreducible $A\in\M_{N,N}(\IRmax)$ and all~$v\in\IR^N$, we have 
$ n_{A,v} \leqslant \Bnetwo(G)$.
\end{thm}
\begin{proof}
By Lemma~\ref{lem:transient:invariance} and Lemma~\ref{lem:invariance}, both~$n_{A,v}$ and~$\Bnetwo\big(G(A,v)\big)$ are invariant under substituting~$A$ by~$\overline{A}$. 
We may hence assume without loss of generality that~$\varrho(A)=0$, i.e., $\varrho\big(G(A,v)\big)=0$.

Let~$B=\Bnetwo\big(G(A,v)\big)$ and~$p=p(G)$.
Graph~$G(A,v)$ is nontrivial and strongly connected, because~$A$ is irreducible.
By Lemma~\ref{lem:realizer:ha}, for every node~$i$ and every~$n\geqslant B$, there exists an $\realrem{B}{n,p}$-realizer for~$i$ of length~$n$.
Lemma~\ref{lem:final:step} hence implies that the transient of sequence~$\big(w^n(\ito)\big)_n$ in graph~$G(A,v)$ is at most~$B$.
But, by Lemma~\ref{lem:path:formula}, the transient of~$\big(w^n(\ito)\big)_n$ in~$G(A,v)$ is equal to~$n_{A,v}$.
This concludes the proof.
\end{proof}

\begin{cor}\label{cor:nonexplorative2}
For all irreducible $A\in\M_{N,N}(\IRmax)$ and all~$v\in\IR^N$, we have 
\[ n_{A,v} \leqslant \max \left\{ \frac{\lVert v \rVert + \big(\Delta_\nc(G) - \delta(G)\big) \cdot(N-1) }{ \varrho(G) - \varrho_\nc(G)}  \ ,\ N^2  \right\} \enspace,\]
where~$G=G(A)$.
\end{cor}


\section{Matrix vs.\ System Transients}\label{sec:matrix}

To obtain an upper bound on the  transient  of a max-plus matrix,
	we can follow the same idea as in Sections~\ref{sec:visiting},~\ref{sec:explorative},
	and~\ref{sec:nonexplorative}, and substitute
	$A_{i,j}^{\otimes n}$ and $\Pa^n(i,j,G)$ for 
	$\big(A^{\otimes n}\otimes v\big)_i $ and $\Pa^n(\ito,G)$ in the proofs.
This leads to an upper bound in $O\big((\Delta-\delta)\cdot N^2/(\varrho-\varrho_\nc)\big)$, but gives no hint on the relationships between
	the transient of a max-plus matrix~$A$, and the transients of the max-plus
	systems~$x_{A,v}$.
In this section, we show that up to some constant $\Bunu(G)$, the  transient of matrix~$A$
	is actually equal to the transient of some specific systems~$x_{A,v}$ where 
	$\lVert v \rVert $ is in $O\big((\Delta-\delta)\cdot N^2\big)$.
Combined with the general upper bounds on the system transient established in 
	Theorems~\ref{thm:explorative:nonprimitive},~\ref{thm:nonexplorative}
	or~\ref{thm:nonexplorative2}, this result gives upper bounds on the matrix transient, 
	each of which is  also  in $O\big( (\Delta-\delta)\cdot N^2/(\varrho-\varrho_\nc)\big)$.

First, we derive a general property of strongly connected graphs from
	the definition of exploration penalty. 

\begin{lem}\label{lem:pathpenality}
Let $G$ be a strongly connected graph.
For any pair of nodes $i,j$ of $G$ and any integer  $n \geqslant \EP(G) + c(G) + \CDD(G) -1$,
	there exists a path $\pi$ from $i$ to $j$ such that $n-\ell(\pi) \in \{0, \cdots, c(G)-1\}$.
\end{lem}

\begin{proof}
Let~$i, j$ be any two nodes, and let $\pi_0$ be a simple path from $i$ to $j$.
For any integer $n$, consider the residue $r$ of $n-\ell (\pi_0)$ modulo $c(G)$.
By definition of $\EP(G)$, if $n-\ell (\pi_0) - r \geqslant \EP(G)$, then there
	exists a closed path $\gamma$ starting at node $j$ with length equal to 
	$n-\ell (\pi_0) - r$.
Then, $\pi_0 \cdot \gamma$ is a path from $i$ to $j$ with length $n-r$, where
	$r \in \{0, \cdots, c(G)-1\}$.
The lemma follows since $n-\ell (\pi_0) - r \geqslant \EP(G)$ as soon as 
	$n \geqslant \EP(G) + \CDD(G) + c(G) - 1$.
\end{proof}	

\begin{lem}\label{lem:ep:g}
Let~$A\in\M_{N,N}(\IRmax)$ irreducible, $G=G(A)$, and let~$n$ be any integer
such that $n \geqslant \EP(G)+ c(G) + \CDD(G) -1$.
Then $A_{i,j}^{\otimes n+c(G)}=-\infty$ if and only if $A_{i,j}^{\otimes n} = -\infty$.
\end{lem}
\begin{proof}
It is equivalent to claim that $\Pa^{n+c(G)}(i,j,G)=\emptyset$ if and
     only if $\Pa^n(i,j,G)=\emptyset$ for any integer~$n\geqslant \EP(G)+ c(G) + \CDD(G) -1$.

Suppose $\Pa^{n+c(G)}(i,j,G) \neq \emptyset$, and let  $\pi_0\in\Pa^{n+c(G)}(i,j,G)$. 
By Lemma~\ref{lem:pathpenality}, there exists a path $\pi\in\Pa(i,j,G)$ such that 
	$n=\ell(\pi)+r$ with $r\in\{0,1,\dots,c(G)-1\}$.
Lemma~\ref{lem:Ni,j} implies that~$c(G)$ divides $\ell(\pi_0)-\ell(\pi) = (n+c(G)) - (n-r) = c(G)+r$; hence~$c(G)$ divides~$r$.
Therefore,~$r=0$, i.e., $\ell(\pi) =n$ and thus~$\Pa^n(i,j,G)\neq\emptyset$.

The converse implication is proved similarly.
\end{proof}

We now define the matrix/system bound for graph~$G$ by
\[ \Bunu(G) = 2\CDD(G) + \EP(G) + \max_{H \in \SCC(G_c)} c(H) + \max_{H \in \SCC(G_c)}\, \EP(H) - 1 \]
and let
\[ \mu(A) = \sup \left\{ A_{i,k}^{\otimes n} - A_{i,j}^{\otimes n} \mid i,j,k\in V,\ n\geqslant \Bunu(G(A)),\ A_{i,j}^{\otimes n} \neq -\infty \right\}\enspace. \]
Obviously, $\mu \in \IRmax$.

In the following lemma, we fix a node $j\in V$, and a vector~$v\in\IR^N$ such 
	that
	\begin{equation}\label{cond:mu}
	 \forall k\in V\setminus \{ j\},\ \   v_j - v_k\geqslant\mu  \enspace.
	\end{equation}
Such a vector $v$ exists, and we let $x = x_{A,v}$.

\begin{lem}\label{lem:mu}
Let~$A\in\M_{N,N}(\IRmax)$ irreducible and $G=G(A)$.
For any integer~$n\geqslant \Bunu(G)$, any node~$i$, and any positive integer $p$ 
	that is a multiple of $c(G)$,
     if $x_i(n+p)=x_i(n)$, then $A_{i,j}^{\otimes n+p} =
     A_{i,j}^{\otimes n}$.
\end{lem}
\begin{proof}
Let~$i$ be  any node in $G$, and $n$ be any integer such that $n \geqslant \Bunu(G)$.
Since~$\Bunu(G) \geqslant \EP(G) + c(G) + \CDD(G) -1$,  and~$c(G)$ divides~$p$,
	we derive from Lemma~\ref{lem:ep:g} that $A_{i,j}^{\otimes n+p}=-\infty$ if 
	and only if $A_{i,j}^{\otimes n}=-\infty$.
There are two cases to consider:
\begin{enumerate}
\item $A_{i,j}^{\otimes n}= -\infty$ and $A_{i,j}^{\otimes n+p}= -\infty$.
In this case, 	$A_{i,j}^{\otimes n+p} = A_{i,j}^{\otimes n}$ trivially holds.
\item  $A_{i,j}^{\otimes n}\neq-\infty$, and $A_{i,j}^{\otimes n+p}\neq -\infty$.
Recall that 
	$$x_i(n) = \max \big\{A_{i,k}^{\otimes n} + v_k \ | \ k\in \{1,\cdots,N\}\big\} \enspace.$$
By definition of $\mu$ and $v$,  for any node $k\neq j$,
	\[ A_{i,k}^{\otimes n} - A_{i,j}^{\otimes n} \leqslant \mu \leqslant v_j - v_k \enspace. \]
Since the latter inequalities trivially hold for $k=j$, it follows that 
     \[ x_i(n) = A_{i,j}^{\otimes n} + v_j \enspace.\]
As  $n +p \geqslant n$, we also have
\[ A_{i,j}^{\otimes n +p} = x_i(n+p) - v_j = x_i(n) - v_j = A_{i,j}^{\otimes n} \enspace. \]
Thus  $A_{i,j}^{\otimes n+p} =A_{i,j}^{\otimes n}$  holds in this case.
\end{enumerate}
The lemma follows in both cases.
\end{proof}

\begin{lem}\label{lem:mu:upper:bound}
Let $A\in\M_{N,N}(\IRmax)$ be an irreducible max-plus matrix, and let $G=G(A)$. Then
\begin{equation}
\displaystyle \mu(A) \leqslant \overline{\Delta}(G) \cdot \CDD(G) - \overline{\delta}(G)\cdot \big( 2\CDD(G) + \EP(G) 
		+ \max_{H \in \SCC(G_c)} c(H) - 1 \big)\enspace.\notag
\end{equation}
\end{lem}
\begin{proof}
First observe that $$\mu(A) = \mu(\overline{A}) \enspace .$$ 
Hence each term in the inequality to show  is invariant under substituting $G$ by 
	$\overline{G}$, and we may assume without loss of generality \ that~$\varrho(A) = 0$.
It follows that $$  \delta(G) \leqslant 0  \leqslant \Delta(G) \enspace .$$ 

Then we  prove that 
\begin{equation}\label{eq:mu:upper:bound:upper}
A_{i,k}^{\otimes n} \leqslant \Delta(G)\cdot \CDD(G)\enspace.
\end{equation}
If $A_{i,k}^{\otimes n} = -\infty$, then (\ref{eq:mu:upper:bound:upper}) trivially holds.
Otherwise,  $A_{i,k}^{\otimes n} = \wstar(\hat{\pi})$ for some path $\hat{\pi}\in\Pa^n(i,k,G)$, 
	and $$\wstar(\hat{\pi}) \leqslant \wstar\big(\!\Simp(\hat{\pi}) \big) \leqslant 
	\Delta (G)\cdot \CDD(G)\enspace.$$

We now give a lower bound on~$A_{i,j}^{\otimes n}$ in the case
     that it is finite, i.e., if~$\Pa^n(i,j,G)\neq\emptyset$.
Let~$u$ be a critical node, in critical component~$H$, with minimal
     distance from~$i$ and let~$\pi_1$ be a shortest path from~$i$
     to~$u$.
Further, let~$\pi_2$ be a shortest path from~$u$ to~$j$.
Let $r$ denote the residue of $n - \ell(\pi_1\cdot\pi_2) - \EP(G)$ modulo $c(H)$,
	and let $t= n - \ell(\pi_1\cdot\pi_2) - \EP(G) -r$.
Since~$t\equiv0\pmod{c(H)}$, and 
	$$t \geqslant \Bunu -  2\CDD(G) - \EP(G) - c(H) + 1 \geqslant  \EP(H) \enspace,$$
	there exists a closed path~$\gamma_c$ of length~$t$ in
	component~$H$ starting at node~$u$.
Let $s= \EP(G) + r$; then, $s \geqslant \EP(G)$.
Moreover, $s= n - \ell(\pi_1\cdot\gamma_c\cdot\pi_2)$, and 
	$\pi_1\cdot\gamma_c\cdot\pi_2 \in \Pa(i,j,G)$.
By Lemma~\ref{lem:Ni,j}, it follows that $c(G)$ divides $s$.
Hence there exists a closed path~$\gamma_{nc}$ of length~$s$ starting
     at node~$j$.

Now define~$\pi=\pi_1\cdot\gamma_c\cdot\pi_2\cdot\gamma_{nc}$.
Figure~\ref{fig:lem:mu:upper:bound} shows path~$\pi$.
Clearly,~$\ell(\pi) =n$ and
	$$ \wstar(\pi) \geqslant \delta \cdot( n - t
     ) \geqslant \delta\cdot \big( 2 \CDD(G) + \EP(G) + c(H) -1
     \big) \enspace,$$
and so
	\begin{equation}\label{eq:mu:upper:bound:lower}
		A_{i,j}^{\otimes n} \geqslant  \delta\cdot \big( 2 \CDD(G) + \EP(G) + 
		\max_{H \in \SCC(G_c)} c(H) -1\big) \enspace.
	\end{equation}
Combining~\eqref{eq:mu:upper:bound:upper}
     and~\eqref{eq:mu:upper:bound:lower} concludes the proof.
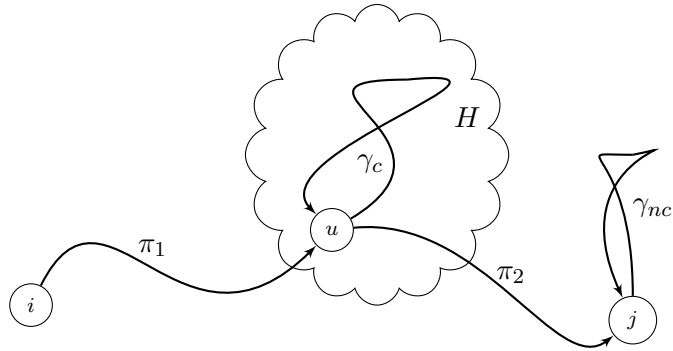
\begin{figure}[ht]
\centering
\begin{tikzpicture}[>=latex']
	\node[shape=circle,draw] (i) at (-4,-1) {$\scriptstyle i$};
	\node[shape=circle,draw] (k) at (0,0) {$\scriptstyle u$};
	\node[shape=circle,draw] (j) at (4,-1.2) {$\scriptstyle j$};
	\draw[thick,->] (i) .. controls (-3,1) and (-2,-2)  .. node[midway,above]{${\pi}_1$}  (k);
	\draw[thick,->] (k) .. controls (2,0.2) and (3,-2)  .. node[midway,above]{${\pi}_2$}  (j);
	\draw[thick] (k) .. controls (2,1.2) and (-1,2)  .. node[near start,left]{$\gamma_{c}$}  (1,2);
	\draw[thick,->] (1,2) .. controls (3,2.2) and (-1,1)  ..  (k);
	\node[cloud, cloud puffs=16, draw,minimum width=3.5cm, minimum height=4cm] at (0.6,1) {};
	\node at (1.8,1.5) {$H$};
	\draw[thick] (j) .. controls (4,1.2) and (3,1)  .. node[near start,right]{$\gamma_{nc}$}  (4,1);
	\draw[thick,->] (4,1) .. controls (5,1.2) and (3,1)  ..  (j);
\end{tikzpicture}
\caption{Path~$\pi$ in proof of Lemma~\ref{lem:mu:upper:bound}}
\label{fig:lem:mu:upper:bound}
\end{figure}
\end{proof}

From Lemma~\ref{lem:mu:upper:bound} with $\CDD(G) \le N-1$ and $c(H) \le N$ we immediately obtain:
\begin{cor}\label{cor:mu:upper:bound}
Let $A\in\M_{N,N}(\IRmax)$ be an irreducible max-plus matrix, and let $G=G(A)$. Then
\begin{equation}
\displaystyle \mu(A) \leqslant
   (\Delta(G) - \delta(G)) \cdot (N-1) + (\varrho(G)-\delta(G))\cdot \big( 2(N-1) + \EP(G)\big)\enspace.\notag
\end{equation}
\end{cor}

For an irreducible matrix $A\in\M_{N,N}(\IRmax)$ and~$j$, with $1\le j\le N$,
we define vector~$v^j$ by:
	$$ \forall k \in V,\ \  v^j_k = \left\{  \begin{array}{ll}
	                                              0 & \mbox{ if } k=j \\
	                                              -\mu(A) & \mbox{ otherwise\enspace. }
	                                          \end{array} \right.$$

\begin{thm}\label{thm:nAvsnAv}
Let~$A\in\M_{N,N}(\IRmax)$ be an irreducible max-plus matrix and $G=G(A)$.
Then $$   n_{A}  \leqslant \
		\max \{ \Bunu(G) \ ,\  n_{A, v^1} \ , \cdots,\ n_{A, v^N} \} \enspace .$$
\end{thm}
\begin{proof}
We easily check that each vector $v^j$ satisfies condition (\ref{cond:mu}), 
	and $\lVert v^j \rVert = \mu(A)$.
Since $c(G)$ divides $c\big(A\big)$, we can apply 
	Lemma~\ref{lem:mu} with all vectors $v^j$ and $p= c\big(A \big)$.
Then, we obtain
	$$ \forall n \in \IN  :\   n \geqslant \
\max \{ \Bunu(G) \ ,\  n_{A, v^1}\ , \cdots,\ n_{A, v^N} \} 
\Longrightarrow 
n \geqslant n_{A}  \enspace ,$$
	which shows 
$$   n_{A}  \leqslant \
\max \{ \Bunu(G) \ ,\  n_{A, v^1} \ , \cdots,\ n_{A, v^N} \} \enspace .$$
\end{proof}

Hence up to  $\Bunu$, the transient of an irreducible matrix $A$ is equal to 
	the transient of some of the systems $(A, v^j)$.
Interestingly, this result could be compared to the equality 
	$$ n_A = \max \{ n_{A, e^1} \ , \cdots,\ n_{A, e^N} \} $$
	established in Section~\ref{subsec:period}, where $e^j$ denotes the vector in $\IRmax^N$
	which is similar to vector $v^j$ except the $j$th component equals to $-\infty$
	instead of $-\mu(A)$.

Combination of Theorem~\ref{thm:nAvsnAv} and Corollary~\ref{cor:mu:upper:bound} finally yields:

\begin{cor}\label{cor:Abound}
Let~$A\in\M_{N,N}(\IRmax)$ be an irreducible max-plus matrix and $G=G(A)$.
Then
\begin{equation}
n_{A}  \leqslant \max \big\{ 3(N-1) + \EP(G) + \max_{H\in\SCC(G_c)} \EP(H)\ ,\ B(G) \big\} \enspace ,\notag
\end{equation}
where $B(G)$ is the minimum of the bounds stated in Corollaries~\ref{cor:explorative:nonprimitive},~\ref{cor:nonexplorative},
and~\ref{cor:nonexplorative2}, with $\lVert v\rVert$
replaced by $(\Delta(G) - \delta(G)) \cdot (N-1) + (\varrho(G)-\delta(G))\cdot \big( 2(N-1) + \EP(G)\big)$.
\end{cor}

\section{Discussion}\label{sec:discussion}

In this section, we discuss the relation to previous work on system and matrix transients and show how to apply our results to the analysis of the Full Reversal algorithm.
In particular, we also obtain a new result on the transient of Full Reversal scheduling on trees.

\subsection{Relation to previous work}\label{subsec:relation:work}

\subsubsection{Even and Rajsbaum}

Even and Rajsbaum~\cite{even:rajs} proved an upper bound on the
transient of $x_{A,v}$ for an irreducible matrix $A \in
\M_{N,N}\big(\IN\cup\{-\infty\}\big)$ and a vector $v \in \IN^N$.
With our notation and~$G=G(A,v)$,
     their bound reads    
\begin{equation}
n_{A,v}^{\mathrm{ER}} = l_0(G) + N + 2N^2\enspace,\label{ERbound:Av}
\end{equation}
where $l_0(G)$ is an upper bound on the length~$n$ of maximum weight
     paths that contain only non-critical nodes.
It thus corresponds to our $\Bcnc(G)$ bound, and is given by,
\begin{equation}
 l_0(G) = \frac{N}{f(G)}\Big(\lVert v\rVert+(\Delta-\delta)\,(N-1) \Big)+(N-1)\enspace,\notag
\end{equation}
where $f(G)$ is defined by
\begin{align}
f(G) &= \inf\Big\{ \ell(\gamma)\varrho(G)- \wstar(\gamma) \mid \gamma\in\CP(G) \wedge\gamma\text{ is non-critical}\Big\}\enspace.\notag\\
\intertext{Since a path that has no critical nodes is non-critical, this can be bounded by}
f(G) &\le N \cdot \big(\varrho(G)-\varrho_\nc(G)\big)\enspace.\notag
\end{align}
Together with Corollary~\ref{cor:n:zero} it thus follows that,
\begin{equation}
 l_0(G) \ge \frac{\lVert v\rVert+(\Delta-\delta)\cdot(N-1)}{\varrho-\varrho_\nc}+(N-1) \ge \Bcnc(G)\enspace.\label{ER:comp1}
\end{equation}

The $N+2N^2$ term in~\eqref{ERbound:Av} corresponds to the second term
     in the maximum of our explorative and repetitive bounds.
Even and Rajsbaum extend realizers that contain critical nodes by
     adding a {\em spanning Eulerian derivative\/} (SED) of
     length~$O(N^2)$ for each critical component visited by the
     original realizer.
An SED is a closed path that visits each node in the critical
     component.
They add SEDs because of the way they reduce paths: Their {\em
     non-critical part\/} construction may disconnect the original
     path by removing collections of critical edges that may be
     combined to a critical closed path.
This is a major difference to our approach: While we, too, reduce
     realizers, our construction does not disconnect paths.
     
With the SED construction Even and Rajsbaum can pump the length of the
     constructed path in multiples of~$c(G_c)$, by adding paths from
     the set of {\em all\/} elementary closed paths of {\em all\/}
     visited critical component, while still maintaining the property
     of being a realizer.
The major difference to our approach here is that we pump the reduced
     realizers' length with either the (i) explorative method, or (ii)
     the repetitive method such that the resulting paths remain
     realizers, where the idea of both (i) and (ii) is that we extend
     the reduced realizers by adding a closed path at a {\em single\/}
     critical node whose subpaths are from a {\em small restricted
     set\/} of elementary closed paths of the node's critical
     component.
In case of (ii) we even restrict the set of elementary closed paths,
     which are used for pumping, to a {\em single\/} elementary closed
     path the critical node lies on.

Not resorting to the SED construction, we obtain upper bounds in which
     the critical path contribution may be linear in the number of
     nodes (cf.\ Section~\ref{subsec:fr} below).

\subsubsection{Hartmann and Arguelles}

Hartmann and Arguelles~\cite{hartmann:arguelles} state the following
bounds
\begin{align}
n_{A,v}^{\mathrm{HA}} &= \max\Bigg\{ \frac{\lVert v\rVert + N\cdot (\Delta-\delta)}{\varrho-\varrho^0} \ ,\ 2N^2 \Bigg\}\enspace,\label{HAbound:Av}\\
n_A^{\mathrm{HA}} &= 2N^2\cdot\frac{\Delta-\delta}{\varrho-\varrho^0}\enspace,\label{HAbound:A}
\end{align}
where $\varrho^0(G(A))$ is defined on the max-balanced reweighted graph~\cite{schneider:schneider} of graph~$G(A)$, in the
     following called~$G_{\max}(A)$, as the supremum of $\varrho' \in
     \IRmax$ such that the subgraph of~$G_{\max}(A)$ induced by the
     edge set of edges of~$G_{\max}(A)$ with weight at
     least~$\varrho'$ has a strongly connected component with only
     non-critical nodes.
In case $\varrho^0=-\infty$, they set $\varrho-\varrho^0 =
     \Delta(G)-\delta(G)$.
Note that the first term in \eqref{HAbound:Av} corresponds to
     our~$\Bcnc(G)$ bound, however, is incomparable with it in
     general.
Hartmann and Arguelles reduce a realizer similar to Even and Rajsbaum, i.e., the reduced path may be disconnected.
To establish connectedness, they add at most~$N$ previously removed closed paths.
They then add additional critical closed paths to arrive at the right residue class modulo~$c(G_c)$.
We have shown in Section~\ref{subsec:improve} that this last step is, in fact, unnecessary and
that the term~$2N^2$ in~$n_{A,v}^{\mathrm{HA}}$
is therefore improvable
     to~$N^2$.
A major difference of our bounds to~\eqref{HAbound:Av}
     and~\eqref{HAbound:A} is that~\eqref{HAbound:Av}
     and~\eqref{HAbound:A} cannot become linear in~$N$.

\subsubsection{Soto y Koelemeijer}

Soto y Koelemeijer~\cite{koelemeijer} presented an upper bound on both
      the transient of system~$x_{A,v}$, in the following denoted by
$n_{A,v}^{\mathrm{SyK}}$, and the transient of matrix~$A$, denoted by~$n_A^{\mathrm{SyK}}$,
for irreducible $A\in \M_{N,N}(\IRmax)$ and $v\in \IR^N$.
With our notation they read, 
\begin{align}
n_{A,v}^{\mathrm{SyK}} &= \max\Bigg\{ \frac{\lVert v\rVert + N\cdot (\Delta-\delta)}{\varrho-\varrho_1} \ ,\ 2N^2 \Bigg\}\enspace,\label{Kbound:Av}\\
n_A^{\mathrm{SyK}} &= \max\Bigg\{ N^2\cdot\frac{\Delta-\delta}{\varrho-\varrho_1} \ ,\ 2N^2 \Bigg\}\enspace,\label{Kbound:A}
\end{align}
where, by setting $G=G(A)$,
\begin{equation}
 \varrho_1(G) = \sup\Bigg\{ \frac{\wstar(\gamma)}{\ell(\gamma)} \mid \gamma\in\CP(G)\ \wedge\ \gamma \text{ is non-critical}\Bigg\}\enspace.\notag
\end{equation}
We start our comparison with the bound $n_{A,v}^{\mathrm{SyK}}$.
By the argument that a path that has no critical nodes is
     non-critical, we obtain, $\varrho_1(G) \ge \varrho_{\nc}(G)$, and
     thereby,    
\begin{equation}
 \varrho(G)-\varrho_1(G) \le \varrho(G)-\varrho_\nc(G)\enspace.
\end{equation}
Thus the first term in \eqref{Kbound:Av} is greater or equal to~$\Bcnc(G)$.
In contrast to the bounds stated in
     Theorems~\ref{thm:explorative:nonprimitive}, \ref{thm:nonexplorative}
     and~\ref{thm:nonexplorative2}, the term~$2N^2$ in
     \eqref{Kbound:Av} prevents the bound from potentially becoming
     linear in~$N$.
Further Corollary~\ref{cor:nonexplorative2} shows that our upper
     bound is strictly less than the upper bound in~\eqref{Kbound:Av}.

\subsubsection{Matrix vs.\ system transients}

In contrast to Soto y Koelemeijer~\cite{koelemeijer} and Hartmann and
     Arguelles~\cite{hartmann:arguelles}, we proved the bound on the
     matrix transient~$n_{A}$ by {\em reduction\/} to the system
     transient~$n_{A,v}$, for properly chosen~$v$, shedding some light
     on how matrix and system transient relate.
With this method we obtain the following bound from
     Corollaries~\ref{cor:Abound} and~\ref{cor:nonexplorative2}, 
\begin{equation}
 n_A \le \max\bigg\{ (N-1)\cdot \frac{2(\Delta-\delta)+(\varrho-\delta)(N+2)}{\varrho-\varrho_\nc} \ ,\ N^2 \bigg\}\enspace,\notag
\end{equation}
which is strictly better than the bounds~$n_A^{\mathrm{HA}}$ and~$n_A^{\mathrm{SyK}}$ with respect to the
     second term, however, in general is incomparable to it with
     respect to the first term.
By conservatively bounding~$\varrho\leqslant \Delta$, we see that our resulting bound on~$n_A$ is close to the bound~$n_A^{\mathrm{SyK}}$ by Soto y Koelemeijer.
Note that, however, our bound can become linear in~$N$, since both
     $\Bms$ and our bounds on~$n_{A,v}$ can become linear in~$N$.

%
%
%

\subsection{Integer matrices}

In the case that all weights of the nontrivial strongly connected e-weighted graph~$G$ are integers, we can derive a lower bound on the term~$\varrho(G) - \varrho_\nc(G)$, which appears in all our upper bounds.
Let~$N$ be the number of nodes in~$G$.
We show that~$1/\big(\varrho(G) - \varrho_\nc(G)\big)$ is in~$O(N^2)$.
This statement is trivial if~$\varrho_\nc(G)=-\infty$; so we assume that~$\varrho_\nc(G)$ is finite.

Let~$\varrho(G)=x/y$ where~$x$ and~$y$ are coprime integers and~$y$ is positive.
Then for every critical closed path~$\gamma$, we have~$\wstar(\gamma)/\ell(\gamma)=x/y$, i.e., $y\cdot\wstar(\gamma)=\ell(\gamma)\cdot x$.
Because~$x$ and~$y$ are coprime, this implies that~$y$ divides~$\ell(\gamma)$.
Hence~$y$ divides all critical closed path lengths, i.e., it divides their greatest common divisor~$\gcd\{ c(H) \mid H\in\SCC(G_c) \}$.
In particular,
\[ y \leqslant \gcd\{ c(H) \mid H\in\SCC(G_c) \}\enspace. \]

If~$\varrho_\nc(G) = z/u$ where~$z$ and~$u$ are coprime integers and~$u$ is positive, then necessarily~$u\leqslant \CF_\nc(G)$, where~$\CF_\nc(G)$ is the maximum path length of elementary closed paths whose nodes are non-critical.
Because~$\varrho(G) - \varrho_\nc(G)>0$, we have
\[ \varrho(G) - \varrho_\nc(G) = \frac{x\cdot u - z\cdot y}{y\cdot u} \geqslant \frac{1}{y\cdot u}\enspace, \]
which implies, by combining the above inequalities, that
\begin{equation}
\frac{1}{\varrho(G) - \varrho_\nc(G)} \leqslant y\cdot u \leqslant \gcd\{ c(H) \mid H\in\SCC(G_c) \}\cdot \CF_\nc(G) \le
(N-N_\nc)\cdot N_\nc \le \frac{N^2}{4} \enspace.\label{eq:int_bound}
\end{equation}
Combination with our bounds on~$n_{A,v}$, for constant~$\lVert v
     \rVert$ and~$\Delta-\delta$, thus yields
     $n_{A,v} = O(N^3)$.

\subsection{Full Reversal routing and scheduling}\label{subsec:fr}

Full Reversal is a simple algorithm on directed graphs used in
     routing~\cite{GB87} and scheduling~\cite{BG89}.
It comprises only a single rule: Each sink reverses all its (incoming)
     edges.
Given an initial graph~$G_0$ with~$N$ nodes, we define a {\em greedy
     execution\/} of Full Reversal as a sequence $(G_t)_{t\ge 0}$ of
     graphs, where~$G_{t+1}$ is obtained from $G_t$ by reversing the
     edges of all sinks in~$G_t$.
As no two sinks in~$G_t$ can be adjacent, $G_{t+1}$ is well-defined.
For each $t\ge 0$ we define the {\em work vector}~$W(t)$ by setting
     $W_i(t)$ to the  number of reversals of node~$i$ until
     iteration~$t$, i.e., the number of times node~$i$ is a sink in
     the execution prefix $G_0,\dots,G_{t-1}$.

Charron-Bost et al.\ \cite{fr:sirocco} have shown that the sequence of
     work vectors can be described as a min-plus linear dynamical
     system: Define a {\em min-plus\/} matrix as a matrix with entries
     in~$\IRmin = \IR\cup \{+\infty\}$.
We, analogously to max-plus, define the matrix multiplication
     \[(A\otimes' B)_{i,j} = \min\{ A_{i,k} + B_{k,j} \mid 1\leqslant
     k\leqslant N\}\enspace.\]
It is $A \otimes' B = -\big((-A) \otimes (-B)\big)$, where
     $(-M)_{i,j}$ is $-M_{i,j}$ for matrix~$M$.
Generalizing a result by Charron-Bost et
     al.~\cite[Corollary~2]{fr:sirocco}, we obtain
\begin{equation}
W(0) = (0,\dots,0) \quad \text{and} \quad -W(t+1) = (-A) \otimes (-W(t))\enspace,\notag
\end{equation}
where
\begin{equation}
A_{i,j} =
\begin{cases}
  0 & \text{ if } (j,i) \text{ is an edge of } G_0\\
  1 & \text{ if } (j,i) \text{ is not an edge of } G_0 \text{ and } (i,j) \text{ is an edge of } G_0\\
+\infty & \text{ otherwise}\enspace.\notag
\end{cases}
\end{equation}

We distinguish two cases that differ in the initial graph~$G_0$: using Full Reversal as a routing algorithm,
     and using Full Reversal as a scheduling algorithm.

\subsubsection{Full Reversal routing}
In the routing case, the initial graph~$G_0$ contains a nonempty set of {\em
     destination nodes}, which are characterized by having a self-loop.
The initial graph without these self-loops is required to be weakly
     connected and acyclic.
It was shown that for such initial graphs, the execution terminates
     (eventually all~$G_t$ are equal), and after termination, the
     graph is destination-oriented, i.e., every node has a path to
     some destination node~\cite{fr:sirocco,GB87}.

In the following we specialize our bounds on the transient of a linear max-plus
     system to obtain bounds on the transient of~$(W(t))_{t\ge
     0}$, i.e., on the termination time~$\theta(G_0)$ of greedy Full Reversal
     routing executions starting from inital graph~$G_0$.
For that we define en-weighted graph~$G = G(-A,0)$ for the max-plus
     matrix~$-A$ obtained from initial graph~$G_0$.
Because $-A_{i,i} = 0$ if $i$ is a destination node, $\varrho(G)=0$.
The set of critical nodes in~$G$ is equal to the set of destination
     nodes in~$G_0$ and each critical component of~$G$ is trivial.
Observe that $-A$ is an integer max-plus matrix with $\Delta_\nc\le 0$
     and $\delta = -1$.
Hence $\varrho_\nc(G) \le -{1}/{N_\nc} \le -{1}/(N-1)$,
     where~$N_\nc$ is the number of non-critical nodes in~$G$.
By Corollary~\ref{cor:n:zero}, $$\Bcnc(G) \le (N-1)^2\enspace.$$
Since for the critical components~$H$ of~$G$, $c(H) = g(H) = 1$, we
     obtain from Theorem~\ref{thm:EP} and
     Corollary~\ref{cor:explorative:primitive}, an upper bound on the
     termination time of Full Reversal routing starting from initial
     graph $G_0$, for $N\geqslant3$, is 
$$\theta(G_0) \leqslant (N-1)^2\enspace,$$
since the second term in the maximum in~$\Bep(G)$ is at most~$2(N-1)$.

If the undirected support of initial graph~$G_0$ without the self-loop
at the destination nodes is a {\em tree}, we can use our bounds to give a new proof
     that the termination time of Full Reversal routing is linear
     in~$N$ \cite[Corollary~5]{fr:sirocco}.
In that case it holds that $\varrho(G)=0$, and either
     $\varrho_\nc(G)=-{1}/{2}$ or $\varrho_\nc(G)=-\infty$.
In both cases we obtain from Corollary~\ref{cor:n:zero},
$$\Bcnc(G) \le 2(N-1)\enspace.$$
In the same way as above, we obtain from Theorem~\ref{thm:EP} and
     Corollary~\ref{cor:explorative:nonprimitive},  
$$\theta(G_0) \le 2(N-1)\enspace,$$
which is linear in~$N$.

\subsubsection{Full Reversal scheduling}

When using the Full Reversal algorithm for scheduling, the undirected support of the weakly connected initial graph~$G_0$ is interpreted as a conflict graph:
nodes model processes and an edge between two processes signifies the existence of a shared resource whose access is mutually exclusive.
The direction of an edge signifies which process is allowed to use the resource next.
A process waits until it is allowed to use all its resources---that is, it waits until it is a sink---and then performs a step, that is, reverses all edges to release its resources.
To guarantee liveness, the initial graph~$G_0$ is required to be acyclic.

Contrary to the routing case, critical components have at least two nodes, because there are no self-loops.
But it still holds that $-A$ is an integer max-plus matrix with $\Delta_\nc\le 0$
and $\delta = -1$.
Hence, by using~\eqref{eq:int_bound} and Corollary~\ref{cor:n:zero}, we get
\[ \Bcnc(G) \le \frac{N^2(N-1)}{4}\enspace, \]
which, by using either Corollary~\ref{cor:explorative:nonprimitive}, \ref{cor:nonexplorative}, or \ref{cor:nonexplorative2} shows that the transient for Full Reversal scheduling is at most cubic in the number~$N$ of processes.
This is an improvement over a result by Malka and Rajsbaum \cite[Theorem~6.4]{malka:rajs} who proved that the transient is at most in the order of~$N^4$.

In the case of Full Reversal scheduling on {\em trees}, it holds that
     $\varrho=-{1}/{2}$, and $\varrho_\nc=-\infty$.
Thus by Corollary~\ref{cor:n:zero},~$\Bcnc(G) = 0$.
Further, $G_\cc = G$ and $c(G)=d(G)=g(G)=2$.
Theorem~\ref{thm:EP} and Corollary~\ref{cor:explorative:nonprimitive}
then imply,
$$\Benp(G) \le 6(N-3)\enspace.$$
The resulting upper bound on the transient of Full Reversal
scheduling on trees is hence linear in~$N$, which was previously unknown.


\section*{Acknowledgments}

The authors would like to thank Sergey Sergeev and Fran\c{c}ois Baccelli for providing useful references.

\end{document}